\documentclass{article} 
\usepackage[preprint]{colm2026_conference}

\usepackage{microtype}
\usepackage{hyperref}
\usepackage{url}
\usepackage{booktabs}
\usepackage{refcount}
\usepackage{mathtools}
\usepackage{lineno}

\definecolor{darkblue}{rgb}{0, 0, 0.5}
\hypersetup{colorlinks=true, citecolor=darkblue, linkcolor=darkblue, urlcolor=darkblue}

\usepackage{latexsym}

\usepackage[T1]{fontenc}

\usepackage[utf8]{inputenc}

\usepackage{amssymb}
\usepackage{makecell}
\usepackage{multirow}
\usepackage{balance} 
\usepackage{amsmath}

\usepackage[ruled]{algorithm2e} 
\SetAlFnt{\small}
\SetAlCapFnt{\small}
\SetAlCapNameFnt{\small}
\SetAlCapHSkip{0pt}
\IncMargin{-\parindent}
\usepackage{tablefootnote}
\usepackage{fvextra} 
\usepackage{algorithmic}
\usepackage{newfloat}
\usepackage{chngcntr} 
\usepackage{amsthm}

\newtheorem{proposition}{Proposition}
\newtheorem{corollary}{Corollary}

\usepackage{listings}
\usepackage{caption} 
\usepackage{float}

\floatstyle{ruled} 
\newfloat{listing}{tb}{lst}{}
\floatname{listing}{Listing}

\counterwithout{listing}{section}

\makeatletter
\@ifundefined{c@listing}{\newcounter{listing}}{} 
\makeatother

\DeclareCaptionStyle{ruled}{labelfont=normalfont,labelsep=colon,strut=off} 
\usepackage{longtable}
\usepackage{array}
\usepackage[table]{xcolor}
\usepackage{ragged2e}
\usepackage{soul}

\definecolor{lemonbg}{RGB}{0,128,0}      
\definecolor{lemontext}{RGB}{255,165,0}  
\definecolor{icecreambg}{RGB}{255,215,0} 
\definecolor{icecreamtext}{RGB}{255,255,255} 

\usepackage{tcolorbox}
\tcbuselibrary{listingsutf8}
\usepackage{capt-of}
\usepackage{subcaption} 
\usepackage{graphicx}
\renewcommand{\bibname}{References}
\renewcommand{\refname}{References}

\newtcolorbox[auto counter]{mybox}[2][]{
  float,
  floatplacement=tbp,
  title={#2},
  colback=gray!10,
  colframe=gray!50,
  coltitle=black,
  #1
}

\renewcommand{\cite}{\citep}

\usepackage{enumitem}
\setlist[itemize]{leftmargin=0.8cm}
\setlist[enumerate]{leftmargin=0.8cm}

\title{Incentive-Aware Multi-Fidelity Optimization for Generative Advertising in Large Language Models}

\author{
\begin{tabular}{l}
\textbf{Jiayuan Liu}$^{1}$, \textbf{Barry Wang}$^{1}$, \textbf{Jiarui Gan}$^{3}$, \textbf{Tonghan Wang}$^{4}$, \textbf{Leon Xie}$^{2}$, \\
\textbf{Mingyu Guo}$^{2}$, \textbf{Vincent Conitzer}$^{1}$ \\
\end{tabular}
\\
\begin{tabular}{l}
$^{1}$ Carnegie Mellon University \quad
$^{2}$ Adelaide University \\
$^{3}$ University of Oxford \quad
$^{4}$ Harvard University \\
\end{tabular}
\\
\\
\begin{tabular}{l}
\small
\texttt{\{jiayuan4,barryw\}@andrew.cmu.edu}, \texttt{jiarui.gan@cs.ox.ac.uk}, \texttt{twang1@g.harvard.edu}, \\
\small
\texttt{leon.xie@student.adelaide.edu.au}, \texttt{mingyu.guo@adelaide.edu.au}, 
\texttt{conitzer@cs.cmu.edu}
\end{tabular}
}

\begin{document}

\ifcolmsubmission
\linenumbers
\fi

\maketitle

\begin{abstract}
Generative advertising in large language model (LLM) responses requires optimizing sponsorship configurations under two strict constraints: the strategic behavior of advertisers and the high cost of stochastic generations. To address this, we propose the \emph{Incentive-Aware Multi-Fidelity Mechanism (IAMFM)}, a unified framework coupling \emph{Vickrey-Clarke-Groves (VCG)} incentives with \emph{Multi-Fidelity Optimization} to maximize expected social welfare. We compare two algorithmic instantiations (elimination-based and model-based), revealing their budget-dependent performance trade-offs. Crucially, to make VCG computationally feasible, we introduce \emph{Active Counterfactual Optimization}, a ``warm-start'' approach that reuses optimization data for efficient payment calculation. We provide formal guarantees for approximate strategy-proofness and individual rationality, establishing a general approach for incentive-aligned, budget-constrained generative processes. Experiments demonstrate that IAMFM outperforms single-fidelity baselines across diverse budgets.
\end{abstract}

\section{Introduction}\label{sec:introduction}

Large Language Models (LLMs) are increasingly deployed as \emph{decision-making and content-generation engines} that simultaneously affect multiple stakeholders. In many such systems, several strategic agents (e.g., advertisers, content creators, or specialized AI assistants) compete to influence a \emph{single} generated response. 
We study this problem through the lens of \emph{generative advertising}, where multiple advertisers compete to steer a single LLM chatbot response shown to the user~\citep{duetting2024mechanism,soumalias2024truthful,dubey2024auctions,feizi2023online}. 
A platform seeks to generate LLM responses that maximize advertiser values and user satisfaction. To do so, it periodically selects a configuration of advertiser ``strength'' parameters that steers the LLM outputs over a subsequent period. Learning an effective selection strategy, however, presents several challenges. First, evaluating any candidate configuration requires repeated LLM queries and downstream evaluations, which incur substantial cost and latency. Second, advertisers are self-interested and may strategically misreport their valuations (e.g., pretending to be unsatisfied with a favorable steered LLM response) to manipulate the chosen strength configuration and, consequently, the generated responses.

To address these challenges, we propose the Incentive-Aware Multi-Fidelity Mechanism (IAMFM). To improve learning efficiency, IAMFM models each advertiser strength configuration as an arm in a multi-fidelity bandit, enabling the platform to trade off evaluation cost against accuracy during exploration. To actively avoid strategic misreporting, IAMFM incorporates a payment rule derived from the Vickrey-Clarke-Groves (VCG) mechanism~\cite{vickrey1961counterspeculation,clarke1971multipart,groves1973incentives} that is computed with respect to expected social welfare. In doing so, IAMFM provides a unified framework that combines the sample efficiency of multi-fidelity optimization (MFO) with incentive-compatibility guarantees from a carefully designed mechanism layered on top.

Our core contributions are as follows:
\begin{itemize}
    \item A new \emph{IAMFM framework} that couples multi-fidelity black-box optimization with an \emph{ex-ante} VCG mechanism. This unified approach maximizes expected social welfare in LLM-based systems while \emph{promoting truthful reporting of valuations} from strategic advertisers.
    \item The design and analysis of \emph{two complementary algorithmic instantiations}---an elimination-based family (\emph{IAMFM-SH/ASH}) and a model-based MFBO variant (\emph{IAMFM-MFBO})---together with empirical evidence showing that their relative advantage depends on the budget.
    \item A \emph{sample-efficient ``warm-start'' procedure, Active Counterfactual Optimization (ACO)}, that reuses main-phase optimization data to estimate VCG counterfactuals, improving the computational efficiency of VCG payment calculations.
    \item \emph{Ex-ante} guarantees showing that, under mild approximation conditions, our mechanism enjoys \emph{approximately strategy-proof}, \emph{approximately individually rational}, and \emph{(approximately) no-subsidy} properties, with explicit violation bounds that depend on optimization and counterfactual errors.
    \item Empirical validation in a realistic LLM simulation environment, providing practical guidance on which IAMFM instantiation to deploy under different computational budgets.
\end{itemize}

While we focus on generative advertising for chatbots, our framework applies to a wider range of tasks characterizing collaborative content-generation with multiple (perhaps conflicting) objectives, such as {AI-assisted scientific discovery} (negotiating trade-offs between efficacy and toxicity in drug design), {personalized education} (balancing students' interests with academic rigor), {LLM platform safety} (optimizing the Pareto frontier between helpfulness and harmlessness), etc. We provide detailed formulations for these broader applications in Appendix~\ref{apdx:use_cases}. 

The remainder of this paper is organized as follows. 
Sec.~\ref{sec:context} details the problem context and challenges. Sec.~\ref{sec:formulation} contains our problem formulation. Sec.~\ref{sec:iamfm_framework} introduces our IAMFM framework and both the model-based and the elimination-based instantiations. Sec.~\ref{sec:vcg_payment} discusses the mechanism rationale with theoretical proofs and acceleration of VCG payment calculation. Sec.~\ref{sec:experiments} evaluates the performance of our method.

\paragraph{Related Work.} Generative advertising in LLMs is an emerging field focusing on auction mechanisms for token-level influence and content positioning \citep{duetting2024mechanism, soumalias2024truthful, dubey2024auctions, feizi2023online}. While multi-fidelity optimization (MFO) \citep{kandasamy2017mfbo, li2018hyperband} provides a standard paradigm for efficient black-box search under budget constraints, its integration with incentive-aligned mechanism design remains largely unexplored. Our IAMFM framework fills this gap by explicitly coupling the sample efficiency of MFO with the strategic robustness of VCG, specifically tailored for the high-cost, stochastic nature of generative systems. A comprehensive literature review is provided in Appendix~\ref{sec:related_work}.

\section{Problem Setup and Challenges}\label{sec:context}

\subsection{Food-Court Generative Advertising Example}
To ground the setting, we consider a \emph{conversational recommendation system} deployed in a shopping mall's food court.\footnote{This could be a kiosk in front of the food court, or a ubiquitous personal agent (e.g., in AR glasses) that delivers an influenced recommendation as the user enters, based on their preferences.}
There are $N$ restaurants (advertisers) who pay to influence the LLM's recommendation as users ask for suggestions (e.g., ``I want something spicy and quick''). The platform chooses an \emph{Influence Configuration}: a vector of \emph{advertising strength levels} $\mathbf{s} = (s_1, \ldots, s_N)$ that controls how strongly each advertiser can steer the generated content.

The platform's goal is to choose $\mathbf{s}$ that maximizes \emph{total social welfare} (advertiser value + user satisfaction). Importantly, higher strength does not always increase welfare: overly promotional responses can reduce user satisfaction even if they increase an advertiser's short-term exposure.
For example, consider how Thai Spice Garden could influence the response to the user's query ``I want something spicy'':
\begin{quote} 
\vspace{-0.1cm}
\textit{Response (Low Strength):} ``For spicy food, Thai cuisine (e.g., Tom Yum soup) and Sichuan cuisine (e.g., hotpot) are popular choices. Two local restaurants for these are Thai Spice Garden and Sichuan Gourmet.''

\textit{Response (High Strength):} ``Hands down, the BEST choice is Thai Spice Garden. Their fiery Panang Curry is exactly what you're craving. It's the city's must-order spicy dish -- skip it and you'll ABSOLUTELY regret it.''
\vspace{-0.1cm}
\end{quote} 
The ``High Strength'' example, with its use of superlatives and pushy, sales-like language, may frustrate the user. 
We model each advertiser's preferences using a ``black-box'' utility function $V_i(\cdot)$ (and similarly define $V_U(\cdot)$ for the user), which maps any generated response $c$ to a utility value. In practice, $V_i$ is implemented by an \emph{LLM agent}, fine-tuned on data from advertiser $i$'s marketing team to act as a proxy for their preferences (similarly for $V_U$ and the fine-tuned LLM agent for user preferences).

This over-saturation trap extends to multi-modal generative advertising. As illustrated in Figure~\ref{fig:visual_saturation}, finding the precise influence configuration is critical to balancing advertising value with image quality. A baseline approach with no advertising focus lacks commercial utility, while balanced parameters yield natural compositions that prioritize one brand. However, extreme settings force the generative model into ``semantic bleeding'', literally fusing the brands together. This visual collapse reinforces why simply maximizing advertiser strength destroys user utility, necessitating a mechanism that optimizes for collective expected welfare.
\begin{figure*}[ht]
    \centering
    \begin{subfigure}[b]{0.23\linewidth}
        \includegraphics[width=\linewidth]{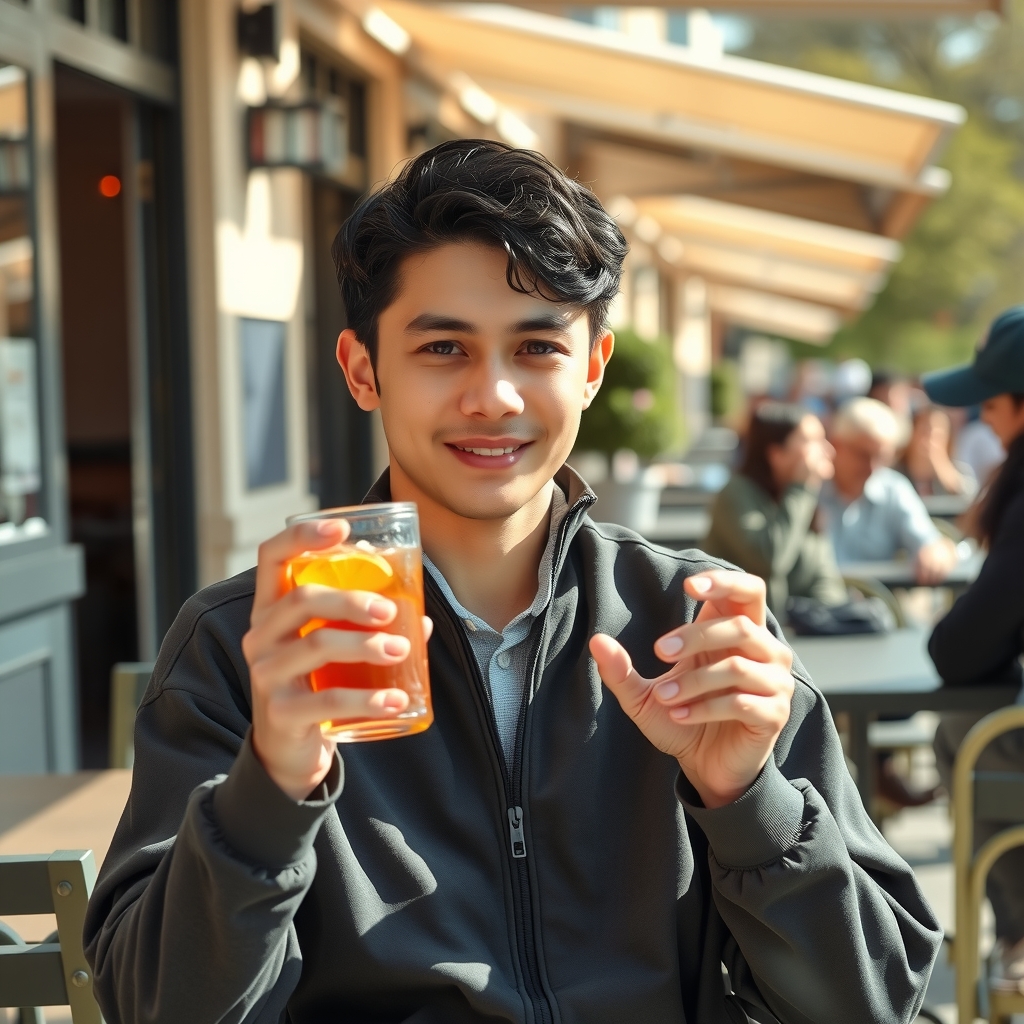}
        \caption{$s_A=s_B={\fontsize{8pt}{10pt} \text{Low}}$}
        \label{fig:visual_baseline}
    \end{subfigure}
    \hfill
    \begin{subfigure}[b]{0.23\linewidth}
        \includegraphics[width=\linewidth]{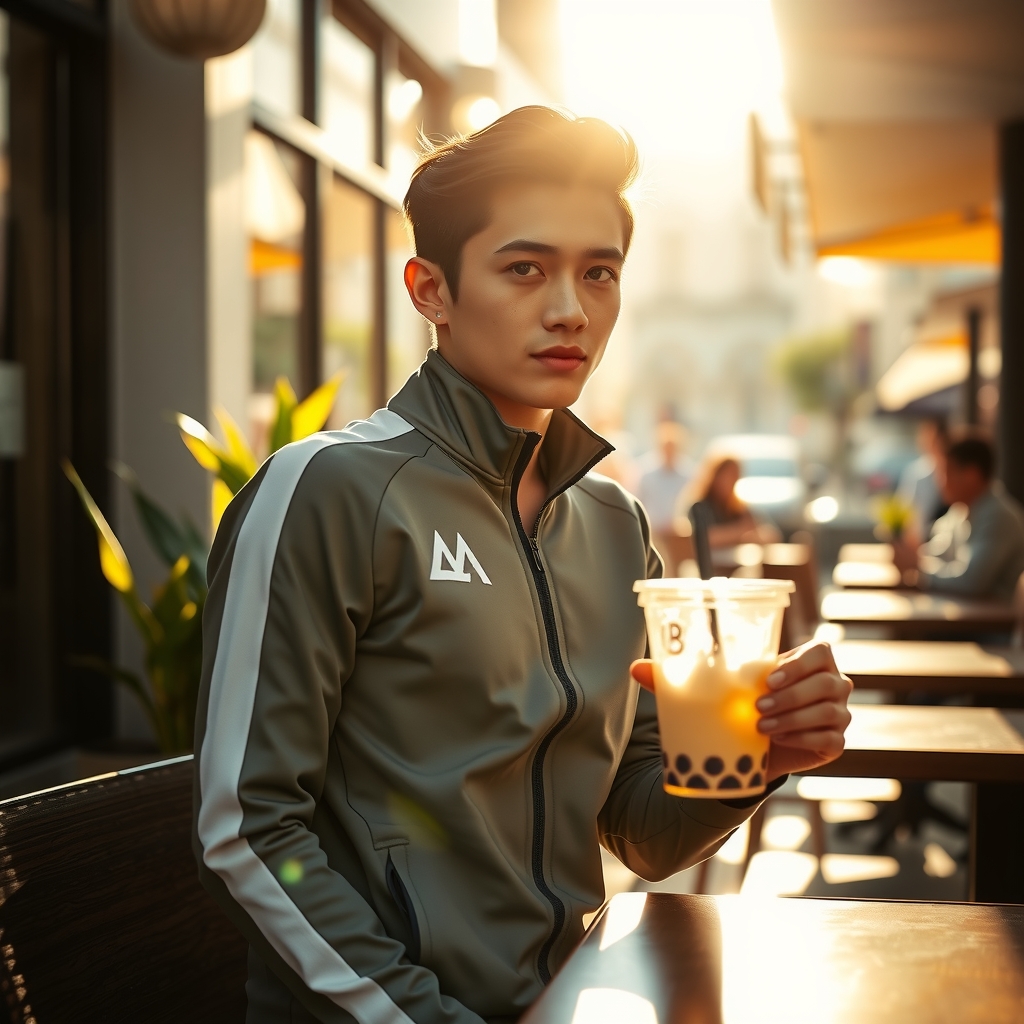}
        {\captionsetup{font=footnotesize}
        \caption{$s_A={\fontsize{8pt}{10pt} \text{Max}}, s_B={\fontsize{8pt}{10pt} \text{Low}}$}}
        \label{fig:visual_a}
    \end{subfigure}
    \hfill
    \begin{subfigure}[b]{0.23\linewidth}
        \includegraphics[width=\linewidth]{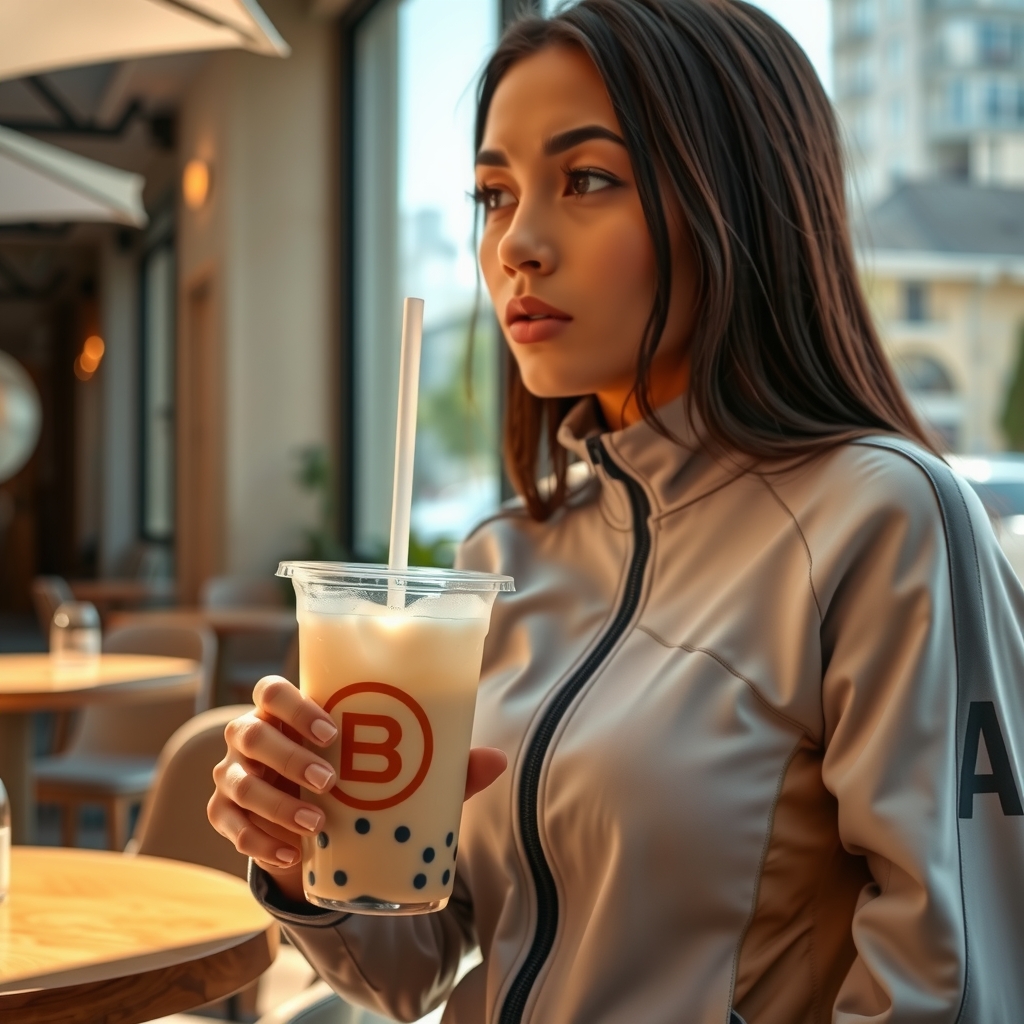}
        \caption{$s_A={\fontsize{8pt}{10pt} \text{Low}}, s_B={\fontsize{8pt}{10pt} \text{Max}}$}
        \label{fig:visual_b}
    \end{subfigure}
    \hfill
    \begin{subfigure}[b]{0.23\linewidth}
        \includegraphics[width=\linewidth]{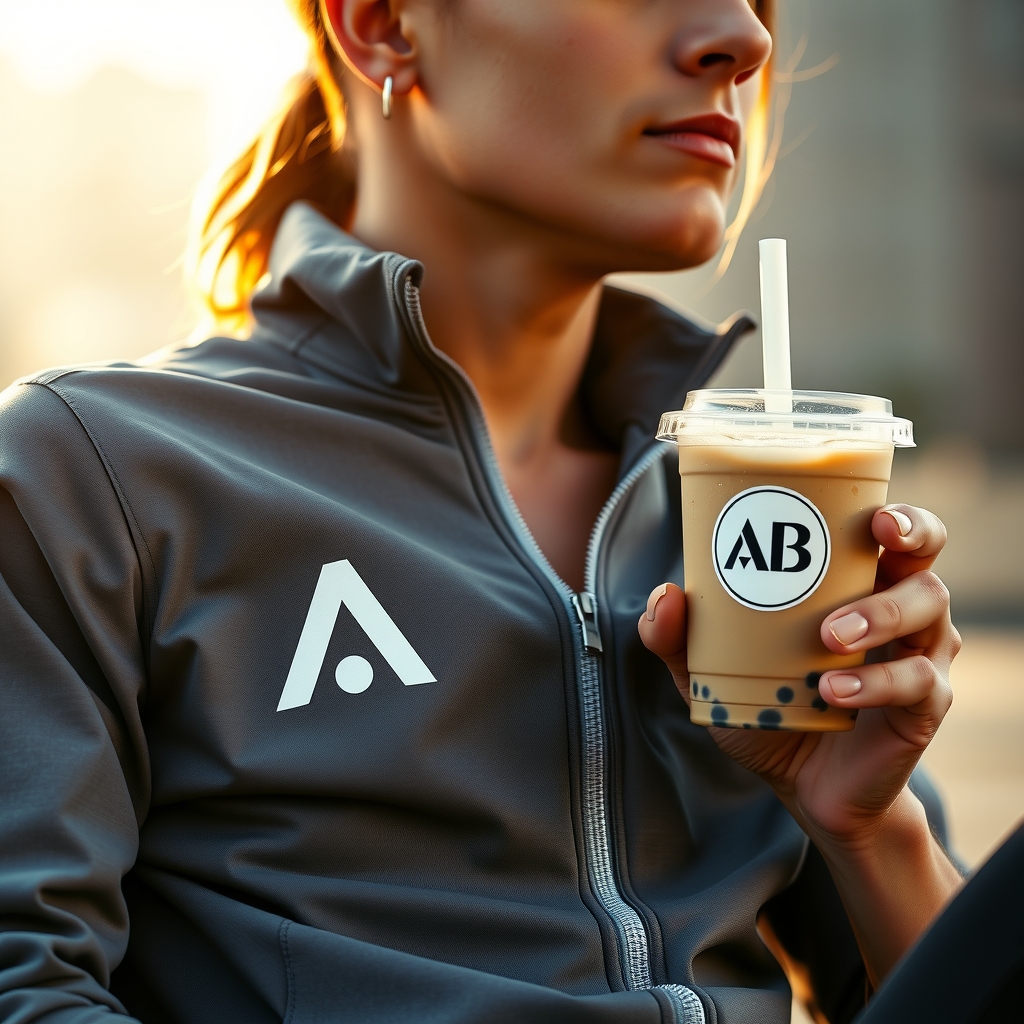}
        \caption{$s_A=s_B={\fontsize{8pt}{10pt} \text{Max}}$}
        \label{fig:visual_c}
    \end{subfigure}
    
    \vspace{-0.1cm} 
    \caption{\textbf{Visual Over-Saturation in Generative Advertising } (advertisers: Alpha sports jacket and Beta boba). (a) Baseline produces a natural image lacking commercial utility (unbranded cup and jacket). Balanced configurations naturally prioritize a single brand: (b) sharp focus on Alpha logo with a blurred Beta branded cup; (c) prominent Beta logo, pushing Alpha to the periphery. Conversely, extreme parameters (d) cause catastrophic structural collapse, hallucinating an amalgamated ``AB'' logo onto the boba cup.}
    \vspace{-0.5cm} 
    \label{fig:visual_saturation}
\end{figure*}

\subsection{The Optimization Objective}
With these utility functions, the platform must find an optimal influence configuration $\mathbf{s}^*$. However, the value of a configuration $\mathbf{s}$ is usually stochastic; it depends on the specific user $u$ and the LLM's random generation. Thus, we evaluate strategies via \emph{sampling}: one must generate many responses across many different user contexts and average the results. The platform's objective is to find a configuration $\mathbf{s}^*$ that maximizes the \emph{expected social welfare}.

\subsection{Key Challenges}
The selection process is constrained by three major challenges:
\begin{enumerate}
    \item \emph{Economic Misalignment:} The advertisers are self-interested. Their interests may not align with the platform's goal of maximizing social welfare, which creates an incentive for them to misreport their true valuations.
    \item \emph{Prohibitive Computational Cost:} The \emph{simulation} (i.e., repeatedly generating LLM responses) required to accurately estimate the expected social welfare is expensive for even \emph{one} configuration. Performing this for \emph{every} configuration in the vast search space is computationally intractable.
    \item \emph{Latency and Synchronization Constraints:} Users expect a response in sub-seconds. This strict requirement makes complex per-query optimization hard to achieve. Furthermore, a real-time bidding architecture would be vulnerable to the ``straggler problem'', where the system's total latency is bottlenecked by the slowest responding LLM proxy agent. If any single agent lags or times out, the entire user experience breaks. 
\end{enumerate}

The cost and latency constraints force the platform to adopt an \emph{offline} optimization paradigm. The expensive search for the optimal configuration must be run periodically (e.g., daily or hourly) to adapt to changing conditions. The single winning configuration from this offline run is then fixed and deployed for the next period (e.g., the next day or next hour), allowing it to serve subsequent user requests with negligible latency.
This fixed deployment strategy gives the optimization objective: the platform's goal is to maximize the expected social welfare aggregated over the entire series of subsequent customers. In Sec.~\ref{sec:formulation}, we proposed a formalized learning problem incorporating multi-fidelity, multi-armed bandit optimization and a mechanism design layer on top.

\section{Problem Formulation}
\label{sec:formulation}

To tackle the computational bottleneck and strategic misalignment identified in Section~\ref{sec:context}, we frame the platform's decision as a \emph{multi-fidelity multi-armed bandit (MAB)} problem. This formulation allows us to explicitly trade off evaluation cost against precision. 

We consider a generative advertising auction with $N$ advertisers. The platform must select an \emph{influence configuration} (or \emph{action}) $\mathbf{s}$ from an action space $\mathcal{A} = S^N$. This action is a vector of \emph{strength levels} $\mathbf{s} = (s_1, \ldots, s_N)$, where each $s_i \in S = \{0, 1, \ldots, k\}$. Each strength level $s_i$ quantifies an advertiser's prominence in the output: $s_i=0$ represents total exclusion, while $s_i=k$ triggers highly promotional language. As we will see, $s_i=0$ is essential for the VCG payment calculation to achieve incentive alignment.

\paragraph{Our Objective: Expected Social Welfare}
The platform operates in short, repeated time windows (e.g., hourly to account for user distribution change). Within each window, we aim to identify the optimal action $\mathbf{s}^*\in \mathcal{A}$ that maximizes the expected social welfare (the sum of advertisers' values plus the user's value), denoted by $W(\mathbf{s})$, across a distribution of users $\mathcal{U}$. Note that $\mathcal{U}$ does not merely represent simple demographic categories, but the vast, high-dimensional space of possible user queries and conversation histories.

\paragraph{Multi-Fidelity MAB Formulation}
To optimize this process, we model the task as a multi-fidelity MAB problem. Each unique action $\mathbf{s} \in \mathcal{A}$ is an \emph{arm} whose latent value is the expected welfare of a complete LLM response generated with influence configuration $\mathbf{s}$. 
We introduce $F$ distinct \emph{fidelity levels} $f \in \{1, \ldots, F\}$.\footnote{While our primary focus throughout this paper is on text-based generative advertising, where fidelity is linked to text generation length (i.e., the length of the prefix generated), we later demonstrate that our framework naturally extends to multi-modal generative advertising (image generation) by defining fidelity levels in terms of image resolutions.}
A lower fidelity level $f$ generates a shorter response prefix $c_f$, e.g., the first $100f$ (or $30 \cdot 2^{f-1}$ or any increasing function of $f$) tokens, which can be ``upgraded'' to a higher fidelity $f' > f$ by continuing generation. Each level incurs a token generation cost $C(f)$, where $C(1) < \dots < C(F)$. The \emph{maximum fidelity level} $F$ corresponds to the generation of a complete LLM response.

The outcome of any action $\mathbf{s}$ is inherently stochastic, governed by two random variables: 
(1) the \emph{user distribution} $\mathcal{U}$, which characterizes the diverse range of user contexts and intents encountered within the current window,
and (2) the \emph{stochastic generation} $c$, drawn from the LLM's generative distribution $P(c \mid \mathbf{s}, f, u)$, which is conditional on the action, fidelity, and context.

\paragraph{Fidelities as Informed Estimators}
By leveraging the Law of Total Expectation, we ensure that the intermediate welfare evaluated at any prefix $c_f$ remains a consistent estimator of the final outcome.
Let $c_F$ denote the complete response at maximum fidelity $F$. Due to the stochastic nature of LLM decoding, we define the value of a prefix, $V_i(c_f)$, as the expected utility of the final response $c_F$, conditioned on the observed prefix $c_f$:
\begin{align*}
V_i(c_f) \coloneqq\mathbb{E}_{c_{f+1:F}}[V_i(c_F) \mid c_f],\ \ 
V_U(c_f) \coloneqq\mathbb{E}_{c_{f+1:F}}[V_U(c_F) \mid c_f].
\end{align*}
We define the realized social welfare $W_{\text{realized}}(c_f)\coloneqq \sum_i V_i(c_f)+V_U(c_f)$. 
Under the Law of Total Expectation, these definitions ensure that the expected welfare at any fidelity $f$ satisfies:
\begin{align*}
\mathbb{E}_{u, c_f}[W_{\text{realized}}(c_f)]
=  \mathbb{E}_{u, c_f} \left[ \mathbb{E}_{c_{f+1:F}} \left[ W_{\text{realized}}(c_F) \mid c_f \right] \right]
=  \mathbb{E}_{u, c_F} [W_{\text{realized}}(c_F)] = W(\mathbf{s}).
\end{align*}

In practice, the platform obtains estimates $\hat{V}_i(c_f)$ and $\hat{V}_U(c_f)$ from \emph{fine-tuned LLM agents}. Each agent functions as a specialized proxy for a specific stakeholder (each advertiser and the user), trained to predict the utility of a full response based on a partial prefix (detailed below in \emph{Arm Pull Process} paragraph). 
While actual agent outputs may introduce \emph{estimation bias} if not perfectly calibrated, our multi-fidelity optimization (MFO) methods only require that the signal $\hat{W}_{\text{realized}}(c_f)$ remains correlated with $W(\mathbf{s})$. 

This converts the problem into a \emph{cost-variance trade-off}: lower fidelities offer computationally cheap but higher-variance signals, while higher fidelities provide more precise estimations at a higher cost. This structure makes our problem ideally suited for MFO methods: leveraging inexpensive, high-variance signals to identify promising candidates within the action space before committing the substantial budget required for their high-precision evaluation at maximum fidelity.

\paragraph{Arm Pull Process (Evaluation via LLM Agents).}
Pulling and evaluating an action $\mathbf{s}$ at fidelity $f$ involves: (1) sampling a user context $u \sim \mathcal{U}$, (2) generating prefix $c_f$ with action $\mathbf{s}$ and context $u$, incurring cost $C(f)$, and (3) querying LLM agents to obtain predictive valuations $\hat{V}_i(c_f)$ and $\hat{V}_U(c_f)$. Critically, LLM agents perform this evaluation orders of magnitude faster than human reviewers. This process also allows for \emph{human-in-the-loop} alignment, where humans can periodically audit and correct agent outputs to refine the agents' underlying valuation models.

Our goal is to identify the best influence configuration $\mathbf{s}^*$ given a fixed computational budget $B$:
\begin{equation}
\mathbf{s}^* \coloneqq \arg\max_{\mathbf{s} \in \mathcal{A}} \hat W(\mathbf{s}).
\end{equation}

\section{Incentive-Aware Multi-Fidelity Mechanism}
\label{sec:iamfm_framework}

Our \emph{Incentive-Aware Multi-Fidelity Mechanism (IAMFM)} integrates a \emph{VCG-based incentive layer} with an \emph{MFO-based optimization engine} to address the economic and computational challenges. We now detail its two complementary instantiations.

\subsection{Elimination-Based Instantiations: IAMFM-SH and IAMFM-ASH}
\label{sec:sh_algorithms}

We introduce two elimination-based instantiations that navigate the \emph{cost-variance trade-off} via a multi-stage tournament structure. We employ the classic \emph{Successive Halving (IAMFM-SH)} as a baseline and propose a novel variant, \emph{Adaptive Successive Halving (IAMFM-ASH)}.

Both algorithms begin with the full action space at the lowest fidelity and progressively filter out low-performing arms at each stage. The critical difference lies in how the budget is allocated \emph{within} a stage:
\begin{itemize}
    \item \textbf{IAMFM-SH} allocates the stage budget uniformly across all surviving arms. (See Appendix~\ref{apdx:sh_algorithm} for details and pseudocode).
    \item \textbf{IAMFM-ASH} (Algorithm~\ref{alg:ash}) replaces uniform allocation with a \emph{UCB-based adaptive sampling strategy}. By dynamically prioritizing arms to balance the exploitation of promising candidates with the exploration of uncertain ones, IAMFM-ASH largely enhances robustness to the stochastic noise inherent in LLM generation.
\end{itemize}

\begin{figure}[tb]
    \centering
    \begin{minipage}[t]{0.503\textwidth}
        \begin{algorithm}[H]
            \caption{IAMFM-ASH (Elimination-Based)}
            \label{alg:ash}
            \small 
            \SetAlgoLined
            \SetAlgoNoEnd
            \DontPrintSemicolon
            \KwIn{Action space $\mathcal{A}$, fidelities $1..F$, budget $B$, rate $\eta$}
            \textbf{Initialize:} $\mathcal{A}_1 \leftarrow \mathcal{A}$\;
            \For{$f = 1$ \KwTo $F$}{
                $B_f \leftarrow$ allocate stage budget from $B$\;
                \emph{Intra-Stage Adaptive Sampling:}\;
                \For{$t = 0$ \KwTo $B_f - 1$}{
                    $\mathbf{s}_{t+1} \gets \arg\max\limits_{\mathbf{s} \in \mathcal{A}_f} \left( \hat{\mu}_{t}(\mathbf{s}, f) + \beta_t \frac{\hat{\sigma}_{t-1}(\mathbf{s}, f)}{\sqrt{n_{t-1}(\mathbf{s}, f)}} \right)$\;
                    Sample $u \sim \mathcal{U}$, generate $c_f$ with $\mathbf{s}_{t+1}$, update stats\;
                }
                \vspace{0.1cm}
                \If{$f < F$}{
                    \textbf{Elimination:} $\mathcal{A}_{f+1} \leftarrow$ Top $1/\eta$ fraction of arms in $\mathcal{A}_f$ by mean $\hat{\mu}(\mathbf{s}, f)$\;
                }
            }
            \Return $\mathbf{s}^* = \arg\max_{\mathbf{s} \in \mathcal{A}_F} \hat{\mu}(\mathbf{s}, F)$
        \end{algorithm}
    \end{minipage}
    \hfill 
    \begin{minipage}[t]{0.48\textwidth}
        \begin{algorithm}[H]
            \caption{IAMFM-MFBO (Model-Based)}
            \label{alg:mfbo_compact}
            \small 
            \SetAlgoLined
            \SetAlgoNoEnd 
            \DontPrintSemicolon 
            \KwIn{Action space $\mathcal{A}$, fidelities $1..F$, budget $B$, reserve $R$}
            \textbf{Initialize:} Train GP surrogate on initial random seed samples $\mathcal{D}_0$\;
            \While{$\text{Cost}(\mathcal{D}_t) < B - R$}{
                Update exploration weight $\beta_t$ via decaying schedule\;
                Identify feasible set: \\
                $\mathcal{F}_t \leftarrow \{(\mathbf{s}, f) \mid \text{Cost}(\mathbf{s}, f) \le B_{\text{remain}} - R\}$\;
                \lIf{$\mathcal{F}_t = \emptyset$}{\textbf{break}} 
                Select $(\mathbf{s}_t, f_t) \leftarrow \arg\max_{(\mathbf{s}, f) \in \mathcal{F}_t} \text{UCB}_t(\mathbf{s}, f)$\;
                Evaluate $(\mathbf{s}_t, f_t)$ to obtain welfare $y_t$\;
                Update dataset $\mathcal{D}_{t+1} \leftarrow \mathcal{D}_t \cup \{(\mathbf{s}_t, f_t, y_t)\}$ and GP posterior\;
            }
            \Return $\mathbf{s}^* = \arg\max_{\mathbf{s} \in \mathcal{A}} \hat \mu_T(\mathbf{s}, F)$
        \end{algorithm}
    \end{minipage}
\end{figure}

\subsection{Model-Based Instantiation: IAMFM-MFBO}
\label{sec:iamfm_mfbo}

Unlike elimination strategies, \emph{IAMFM-MFBO} (Algorithm~\ref{alg:mfbo_compact}) utilizes Multi-Fidelity Bayesian Optimization (MFBO) to build a global \emph{Gaussian Process (GP) surrogate model}, $W(\mathbf{s}, f)$, over the joint configuration-fidelity space.
Specifically, we model the unknown welfare function $W$ as a sample from a Gaussian Process. Under the standard GP assumption, conditioned on observed history, the value of $W(\mathbf{s}, f)$ follows a normal distribution characterized by a predictive mean $\hat \mu_t$ and variance $\hat \sigma_t^2$.

\textbf{Gaussian Process Proxy.}
The GP's covariance structure enables \emph{information transfer}: observations at cheap, low-fidelity levels ($f < F$) update the model's posterior belief (mean $\hat \mu$ and uncertainty $\hat \sigma$) regarding expensive, high-fidelity outcomes globally. This allows the platform to use noisy, inexpensive signals to identify promising regions of the action space without exhaustive high-fidelity sampling. See Appendix~\ref{apdx:mfbo_algorithm_kernel} for kernel specifications.

\textbf{Cost-Aware Search.}
To balance exploration and exploitation, we employ a budget-aware Upper Confidence Bound (UCB) acquisition function: $\text{UCB}_t(\mathbf{s}, f) = \hat \mu_t(\mathbf{s}, f) + \sqrt{\beta_t} \cdot \hat \sigma_t(\mathbf{s}, f)$.
At each step, the algorithm selects the pair $(\mathbf{s}, f)$ with the highest UCB score from the set of {feasible} options.
The final recommendation is the arm with the highest posterior mean at max fidelity: $\mathbf{s}^* = \arg\max_{\mathbf{s}} \hat \mu_T(\mathbf{s}, F)$.
The full details are in Appendix~\ref{apdx:mfbo_algorithm}.

\textbf{Theoretical Guarantee.} 
Under standard regularity assumptions, IAMFM-MFBO achieves a sub-linear simple regret bound of $r_T = \mathcal{O}(\sqrt{\gamma_T/T})$, where $\gamma_T$ is the maximum information gain about the function $\hat \mu$ after $T$ steps, guaranteeing convergence to the optimal expected welfare. (Full proposition and proof in Appendix~\ref{apdx:mfbo_theory_proofs}).

\vspace{-0.5cm}
\section{VCG Payments and Incentives}
\label{sec:vcg_payment}

\subsection{Mechanism Design and Payment Rule}
In a decentralized ecosystem, advertisers are strategic agents who may misreport their valuations to manipulate the outcome. To ensure that truthful reporting is a dominant strategy, we employ the \emph{VCG mechanism}. This aligns individual incentives with the global objective by charging each agent the \emph{opportunity cost} they impose on others.

Let $\mathbf{s}^*$ be the optimal strategy found by our algorithm. The payment $p_i$ for agent $i$ is calculated as: 
\begin{equation}
    p_i = \underbrace{\max_{\mathbf{s} \in \mathcal{A}, s_i=0} \hat W(\mathbf{s})}_{\hat W^{\neg i}_{\text{opt}}} - \sum_{j \neq i} \hat V_j(\mathbf{s}^*).
\end{equation}
Here, $\hat W^{\neg i}_{\text{opt}}$ represents the estimated \emph{counterfactual optimum}---the maximum welfare achievable if agent $i$ had not participated. 
Under this rule, an agent's utility becomes $U_i = \hat{V}_i(\mathbf{s}^*) - p_i = \hat{W}(\mathbf{s}^*) - \hat{W}^{\neg i}_{\text{opt}}$.
In the ideal case where estimates are perfect ($\hat{W}=W, \hat{V}=V$), maximizing $U_i$ is strictly equivalent to maximizing global welfare $W(\mathbf{s}^*)$ because $\hat{W}^{\neg i}_{\text{opt}}$ is independent of $i$'s bids.
However, since our budget-constrained algorithm only yields estimated values finding an $\varepsilon$-approximate solution, these strict properties relax to \emph{approximate guarantees}. 

\begin{proposition}[Bounded incentive to misreport under approximate welfare optimality]\label{prop:vcg_sp}
Suppose the influence configuration $\mathbf{s}^*$ chosen by our algorithm (when all agents are truthful, $b=V$) is $\varepsilon$-approximate with respect to the true optimal \emph{expected} welfare, in the sense that
\[W(\mathbf{s}^*) \ge W(\mathbf{s}) - \varepsilon \quad \text{for all strategies } \mathbf{s} \in \mathcal{A} \ ,\] then for every single agent $i \in N$ and for every unilateral misreport $b_i'$, their \emph{expected} utility $U_i$ satisfies:
\[U_i(b_i', V_{-i}) \le U_i(V_i, V_{-i}) + \varepsilon \]
where $U_i(b) = V_i(\mathbf{s}^*(b)) - p_i(\mathbf{s}^*(b))$ is agent $i$’s true expected utility under reports $b$.
\end{proposition}
Our mechanism also guarantees \emph{approximate individual rationality} (truthful participation yields expected utility $U_i \ge -\varepsilon$, which ensures that advertisers effectively never lose money by participating, encouraging them to join the auction), and \emph{approximate no-subsidy} (the mechanism requires a deficit subsidy of at most $n\varepsilon$, which guarantees that the platform does not need to inject significant external funds to keep the system running, ensuring economic viability). Formal statements and proofs are in Appendix~\ref{apdx:rationale_proofs}.

\subsection{Active Counterfactual Optimization (ACO)}
\label{sec:active_vcg}
A major bottleneck in VCG is computational cost: calculating payments requires solving $N$ additional ``cold-start'' optimization problems to find each $W^{\neg i}_{\text{opt}}$.
To solve this, we introduce the \emph{Active Counterfactual Optimization (ACO)} framework. The core idea is to ``warm-start'' the counterfactual searches by reusing information gathered during the main optimization phase. It substantially reduces the cost of payment calculation compared to naive methods. 
\begin{itemize}
    \item For \emph{IAMFM-MFBO}, we reuse the trained Gaussian Process from the main optimization phase as a \emph{global prior}, allowing the counterfactual search to converge with negligible additional samples.
    \item For \emph{IAMFM-ASH}, we reuse the \emph{sample history}  from the main tournament phase to initialize the counterfactual search.
\end{itemize}
Detailed algorithms for ACO are provided in Appendix~\ref{apdx:aco_details}. 
While both warm-start strategies outperform cold-start optimization, ACO for IAMFM-MFBO is substantially more sample-efficient. The ``global prior'' captures structural trends across the entire action space, whereas in ACO for IAMFM-ASH, the ``local reuse of sample history'' is restricted to the specific points sampled by chance. Consequently, as the number of advertisers $N$ grows, the budget savings from IAMFM-MFBO's global warm-start become increasingly dominant, as shown in our experiments in Sec.~\ref{sec:experiments}.

\section{Experimental Evaluation}
\label{sec:experiments}

We evaluate our proposed instantiations of the IAMFM framework against standard single-fidelity baselines. Our framework assesses performance across a wide range of computational budgets, reflecting the real-world constraints of pre-calculating strategies in LLM-based advertising systems.

\begin{figure*}[t]
\centering
\includegraphics[width=1\linewidth]{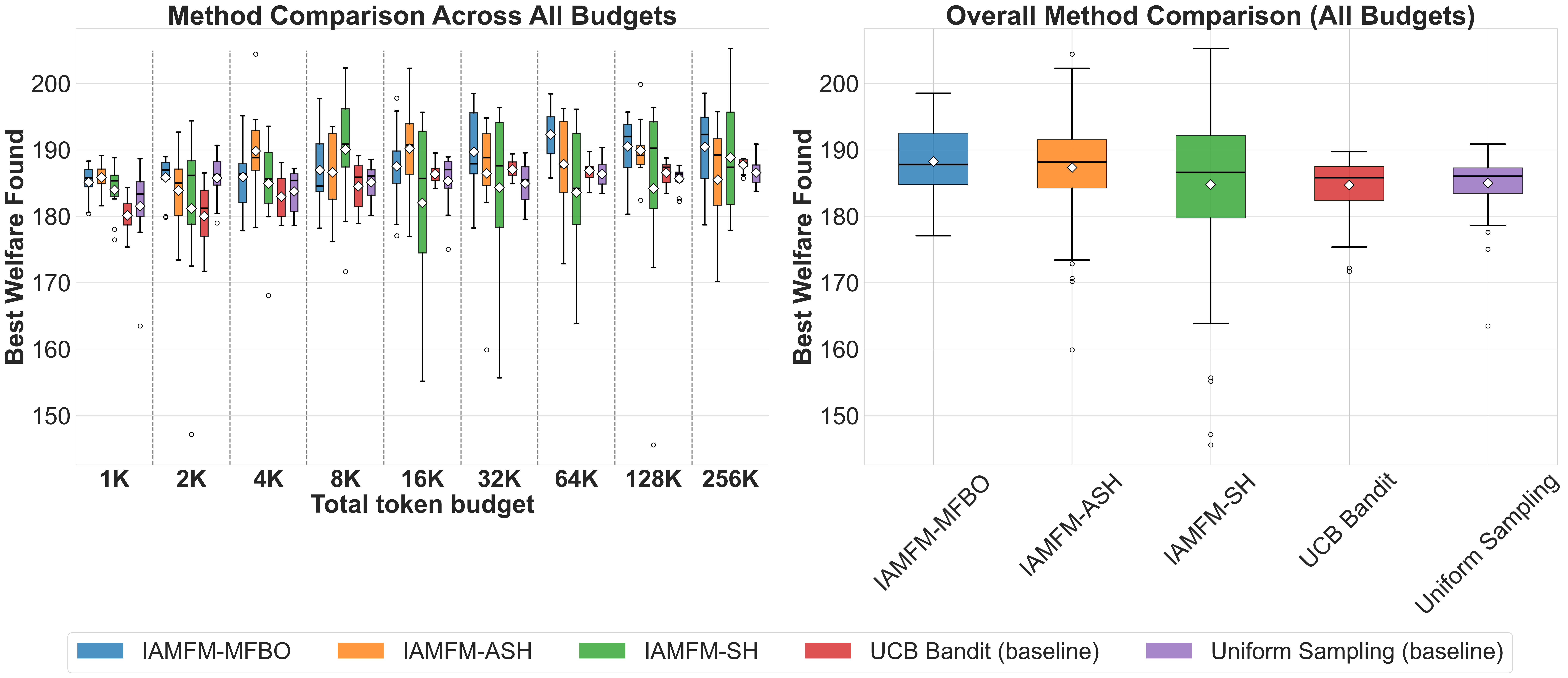}
\caption{Comparison of welfare across all methods. Left: Box plots stratified by token budget summarize performance, showing the median (center line), mean (diamond), interquartile range (box), and outliers (points). All reported statistics include outliers. The IAMFM methods outperform the baselines and reveal a token-budget-dependent shift in the optimal method. Right: Aggregated performance across all budgets. Our proposed multi-fidelity instantiations (IAMFM-MFBO, IAMFM-ASH) show better median performance, while IAMFM-MFBO demonstrates a clear advantage at higher budgets.}
\label{fig:unified_comparison}
\end{figure*}

\textbf{Simulation Environment.}
We evaluate our algorithms in a simulated food court recommendation system with two strategic advertisers and five distinct user personas (e.g., health-conscious, budget-conscious; see Appendix~\ref{apdx:exp_details} for details). 
The action space is discrete, consisting of all 25 strength combinations ($s_A, s_B \in \{0, \dots, 4\}$).
The objective is to find the configuration that maximizes the \emph{expected social welfare} (sum of advertisers' and user's utilities), which is estimated by LLM evaluators.

We define four fidelity levels with corresponding token costs: {F1} (30 tokens), {F2} (60 tokens), {F3} (120 tokens), and {F4} (240 tokens). 
We compare our multi-fidelity instantiations against two standard single-fidelity baselines: {UCB Bandit} and {Uniform Sampling} (both operating at max fidelity). 
Performance is evaluated across 8 computational budgets (1k to 128k tokens) with 10 independent trials each.

\textbf{Results and Analysis.}
Our experiments reveal a distinct performance trade-off governed by the available computational budget, as shown in Figure~\ref{fig:unified_comparison} and Table~\ref{tab:budget_regime_performance}.

In the \emph{Low-Budget Regime} ($\le 16\text{k}$ tokens), both {IAMFM-ASH} and {IAMFM-MFBO} dominate, achieving statistically significant gains over single-fidelity baselines (e.g., ASH vs. UCB: $p < 0.0001$). At this stage, data is scarce, and the primary value comes from aggressively pruning poor configurations using cheap, low-fidelity samples. Consequently, the heuristic elimination of ASH is as effective as the model-based MFBO ($p = 0.38$).
However, the landscape shifts in the \emph{High-Budget Regime} ($\ge 32\text{k}$ tokens). Here, {IAMFM-MFBO} emerges as the statistically superior strategy, outperforming all baselines and, critically, surpassing IAMFM-ASH ($p = 0.036$). As the budget allows for more pulls, the GP surrogate's ability to model the global landscape becomes the deciding factor, finding precise optima that heuristic elimination strategies miss. See Appendix~\ref{apdx:stat_details} for more details. 

Another finding is the superior stability of our model-based approach. The standard elimination baseline ({IAMFM-SH}) suffers from high variance (std. dev. $10.1$--$11.7$), confirming that rigid elimination schedules are brittle: one unlucky low-fidelity sample can permanently discard the optimal arm. 
In contrast, {IAMFM-MFBO} achieves the \emph{lowest variance} among all methods (std. dev. $5.0$--$5.5$). By aggregating information across the entire action-fidelity space, the global GP surrogate effectively recovers the latent signal from the stochastic LLM generations, preventing the algorithm from being misled by noisy samples. This makes MFBO the most robust choice for production deployment.

We observe that as budgets grow very large ($>128\text{k}$ tokens), the performance gap between our methods and brute-force baselines (Uniform/UCB) narrows. This is expected: with infinite resources, even random sampling will eventually find the optimum. The significance of IAMFM is that it achieves near-optimal welfare with {substantially fewer tokens}, operating efficiently in the realistic regime where brute-force is cost-prohibitive.

\textbf{VCG Payment Calculation Efficiency.}
We evaluated the efficiency of our \emph{Active Counterfactual Optimization (ACO)} (details in Appendix~\ref{apdx:aco_analysis}). \emph{ACO for MFBO} proved decisively superior: by reusing the global GP prior, it achieved near-zero estimation error for the counterfactual optimum with \emph{zero additional budget}, reducing VCG computational overhead by $>99\%$ compared to cold-start baselines. In contrast, \emph{ACO for ASH} offered only modest improvements, highlighting that global surrogate models are vastly better at making truthful mechanisms computationally feasible than local sample-history reuse.

\begin{table}[tb]
\centering
\vspace{-0.8cm}
\caption{Aggregated performance (mean welfare $\pm$ std. dev.) stratified by budget regime. Methods are sorted by High-Budget performance. \textbf{Bolded} are proposed methods.}
\label{tab:budget_regime_performance}
{\small
\vspace{-0.2cm}
\begin{tabular}{@{}lcc@{}}
\toprule
\multirow{2}{*}{\textbf{Method}} & \multicolumn{2}{c}{\textbf{Mean Welfare $\pm$ Std. Dev.}} \\ \cmidrule(l){2-3} 
 & \textbf{Low-Budget ($\leq 16\text{k}$ tokens)} & \textbf{High-Budget ($\geq 32\text{k}$ tokens)} \\ \midrule
\textbf{IAMFM-MFBO} & \textbf{186.26 $\pm$ 5.02} & \textbf{190.73 $\pm$ 5.53} \\
\textbf{IAMFM-ASH}  & \textbf{187.26 $\pm$ 6.00} & \textbf{187.43 $\pm$ 7.93} \\
UCB Bandit          & 182.78 $\pm$ 4.15          & 187.04 $\pm$ 1.58 \\
Uniform Sampling    & 184.26 $\pm$ 4.62          & 185.87 $\pm$ 2.41 \\
IAMFM-SH            & 184.42 $\pm$ 10.10         & 185.21 $\pm$ 11.68 \\ \bottomrule
\vspace{-0.5cm}
\end{tabular}
}
\end{table}

\textbf{Qualitative Analysis: The Failure of Greedy Strategies.}
While a greedy strategy (na\"{\i}vely optimizing each advertiser's strength independently as if they were the sole sponsor) pushes the LLM into semantic collapse (producing disjointed text or hallucinated logos), IAMFM identifies the optimal strength configuration that maintains structural coherence. A detailed qualitative text-generation case study demonstrating this phenomenon is provided in Appendix~\ref{apdx:example_qualitative_analysis}.

\textbf{Extension to Multi-Modal Generative Advertising (Image Generation).}
To demonstrate IAMFM's generalizability, we extended our evaluation to visual advertising (introduced in Sec.~\ref{sec:context}, Fig.~\ref{fig:visual_saturation}). Using a distilled diffusion model (\texttt{flux-1-schnell}), we defined F1 as a low-cost, low-step ($512\times512$) exploration phase and F2 as a high-fidelity \emph{img2img} upscaling phase, with a Vision-Language Model (\texttt{llama-3.2-11b-vision-instruct}) serving as the evaluator agent. Consistent with our text-based results, {IAMFM-MFBO} successfully navigated this multi-modal cost-variance trade-off, efficiently identifying the optimal visual saturation point and outperforming single-fidelity baselines by $+4.8$ welfare at $B=4000$. 
Table~\ref{tab:image_results} summarizes the performance of the algorithms across various computational budgets. Full multi-modal experimental details are provided in Appendix~\ref{apdx:image_experiment}.

\begin{table}[h]
\centering
\caption{Algorithm performance (Mean Welfare $\pm$ Std. Dev.) on the Image Generation task across 20 independent trials. {IAMFM-MFBO} demonstrates a widening performance gap over baselines as the budget scales.}
\label{tab:image_results}
\resizebox{\columnwidth}{!}{%
\begin{tabular}{@{}lccccc@{}}
\toprule
\textbf{Budget} & \textbf{IAMFM-MFBO} & \textbf{IAMFM-ASH} & \textbf{IAMFM-SH} & \textbf{UCB Bandit} & \textbf{Uniform Sampling} \\ \midrule
500 & 231.25 $\pm$ 5.29 & \textbf{233.12} $\pm$ 1.97 & 231.24 $\pm$ 3.64 & 225.41 $\pm$ 2.14 & 228.95 $\pm$ 5.45 \\
1,000 & \textbf{232.07} $\pm$ 5.86 & 230.56 $\pm$ 5.79 & 231.96 $\pm$ 7.39 & 229.12 $\pm$ 2.87 & 230.80 $\pm$ 5.08 \\
2,000 & \textbf{233.29} $\pm$ 7.93 & 231.99 $\pm$ 5.99 & 233.19 $\pm$ 5.40 & 229.67 $\pm$ 4.56 & 230.56 $\pm$ 5.65 \\
4,000 & \textbf{238.24} $\pm$ 6.85 & 231.43 $\pm$ 5.20 & 233.41 $\pm$ 5.04 & 233.14 $\pm$ 5.23 & 232.25 $\pm$ 4.85 \\
8,000 & \textbf{240.26} $\pm$ 5.92 & 232.88 $\pm$ 5.79 & 231.48 $\pm$ 6.79 & 235.58 $\pm$ 2.16 & 229.98 $\pm$ 5.05 \\ \bottomrule
\vspace{-0.5cm}
\end{tabular}
}
\end{table}

\section{Conclusion}
\label{sec:conclusion}

We proposed the \emph{Incentive-Aware Multi-Fidelity Mechanism (IAMFM)}, integrating VCG incentives with multi-fidelity optimization to address the economic and computational constraints of generative advertising. Empirically, we found that the optimal strategy is \emph{budget-dependent}: IAMFM-ASH performs comparably well in low-budget regimes, while IAMFM-MFBO dominates at higher budgets. Additionally, we demonstrated that our ACO framework renders VCG payment calculations computationally feasible by reusing the data and model from the optimization process. 
Theoretically, our framework provides ex-ante guarantees for approximate strategy-proofness and individual rationality. Practically, IAMFM successfully prevents the semantic collapse caused by naive over-promotion, establishing a robust algorithmic and economic foundation for future multi-stakeholder LLM ecosystems.  

Beyond advertising, IAMFM lays a robust algorithmic foundation for any multi-stakeholder LLM ecosystem, paving the way for applications in domains like collaborative content creation and AI-assisted scientific discovery.
Extended discussions on alternative steering mechanisms, broader applications, and user transparency are provided in Appendix~\ref{apdx:extended_discussion}.

\section*{Acknowledgment}
Jiayuan Liu and Vincent Conitzer thank the Cooperative AI Foundation,  Macroscopic Ventures (formerly Polaris Ventures / the Center for Emerging Risk Research) and Jaan Tallinn’s donor-advised fund at Founders Pledge for financial support.

\bibliography{colm2026_conference}
\bibliographystyle{colm2026_conference}

\appendix

\section{Detailed Related Work}
\label{sec:related_work}
This work lies at the intersection of mechanism design for advertising, computationally aware optimization, and the emerging application of LLMs in multi-agent systems.

\paragraph{Influencing LLMs}

We are interested in methods that increase exposure of certain topics and brands, or even produce customizable messages, without any post-generation editing. Influencing modern, transformer-based LLMs can occur at all stages. At post-training, parameter-efficient fine-tuning (PEFT) inserts small modules (adapters like LoRA) that achieve strong effects but require careful composition (e.g., AdapterFusion and LoRA can be used to blend multiple styles or stakeholder goals \citep{houlsby2019parameter,hu2022lora,pfeiffer2021adapterfusion,huang2023lorahub}). At inference, cheap prompting and system messages remain the most agile levers, though they can be brittle or context-hungry \citep{brown2020language}. Plug-and-play decoding steers generation via auxiliary scores or gradients without updating weights at some latency/fluency cost \citep{dathathri2020plug,krause2021gedi,yang2021fudge}. Activation steering with sparse autoencoders edits interpretable, extensible features on-the-fly, offering modular, transparently continuous control knobs on influence strength \citep{obrien2024sae, koriagin2025teach}. By contrast, deliberate curation of duplicated data, full pretraining, RLHF, or DPO hard-wire behaviors \citep{ouyang2022training,rafailov2023direct}, which are impractical for rapid, per-query, multi-stakeholder advertising.

\paragraph{Mechanism Design for Advertising Auctions}

The field of advertising auctions is foundational to our work. Traditional search advertising has been dominated by the Generalized Second-Price (GSP) auction \cite{edelman2007internet,varian2007position}. While widely adopted, GSP is not strategy-proof. The Vickrey-Clarke-Groves (VCG)~\cite{vickrey1961counterspeculation,clarke1971multipart,groves1973incentives} mechanism serves as the primary truthful alternative, guaranteeing that truthful bidding is a dominant strategy for advertisers.

However, these classical models were designed for discrete, ``atomic'' ad slots. They do not address the challenges of generative advertising, where the ``good'' being auctioned is a complex, co-created piece of content whose quality is a result of the combined inputs from multiple strategic agents. Our work extends the principles of VCG to this new, generative domain.

\paragraph{Computation-Aware Mechanism Design}

A significant challenge in mechanism design is the computational complexity of determining outcomes and payments, especially with many agents or a large outcome space. Prior work has addressed this in several ways. One line of research focuses on approximating VCG payments, often through sampling-based methods that can provide unbiased estimates but may suffer from high variance \cite{sandholm2005approximate}. Another line of work on algorithmic mechanism design concerns the joint design of (approximation) algorithms and payment functions that make those algorithms into truthful mechanisms~\cite{Nisan01:Algorithmic}. Yet another direction is automated mechanism design~\cite{conitzer2002complexity,sandholm2003automated}, which leverages optimization~\cite{conitzer2003applications,conitzer2003automated}, machine learning~\cite{RochetNet,GEMNet,shen2018automated,rahme2021permutation,feng2018deep}, and more recently large language models~\cite{liu2025interpretable} to derive optimal or approximately optimal mechanisms. These approaches often relax the requirement of perfect strategy-proofness to achieve computational tractability.

In our approach, rather than first determining an outcome and then struggling to compute payments, we integrate the search for the optimal outcome and the calculation of its associated payments into a single, unified optimization process. Our key contribution is a method to make the VCG calculation itself sample-efficient by reusing information.

\paragraph{Multi-fidelity Bayesian Optimization and Bandits}

To address the expensive, black-box nature of evaluating LLM outputs, we turn to the literature on sequential optimization. Our problem, with its discrete set of strength combinations, can be framed as a multi-armed bandit (MAB) problem. Since evaluations can be performed at varying levels of completeness (i.e., number of tokens), this is specifically a multi-fidelity optimization problem.

State-of-the-art methods for this include non-Bayesian approaches like Successive Halving (SH) and Hyperband \cite{li2018hyperband}, and Bayesian approaches like Multi-fidelity Bayesian Optimization (MFBO) \cite{kandasamy2017mfbo}. We build upon MFBO using a Gaussian Process surrogate model to allocate our limited computational budget. The primary novelty of our work is the application of MFBO not merely as an optimization tool for a single decision-maker, but as the core engine within an incentive-aware mechanism designed for strategic agents.

\paragraph{LLMs in Economic and Multi-Agent Settings}

A growing line of research explores the use of large language models (LLMs) in economic and multi-agent contexts. LLMs have been deployed to simulate human-like behavior as agents in various games \cite{horton2023agents}, and to serve as evaluators of complex, qualitative outcomes. Within this broader direction, and closest to our own work, a particularly active thread concerns auction mechanisms for generative content.

\citet{duetting2024mechanism} introduce a token-by-token bidding framework, in which advertiser agents place bids for insertion slots during text generation, laying the groundwork for applying sequential auction theory to LLMs. Building on this foundation, \citet{soumalias2024truthful} address vulnerabilities such as misreporting and the ``exposure problem,'' where small wording changes can drastically alter advertiser value. Their mechanisms emphasize truthfulness and robustness, often by penalizing deviations from a baseline ``truthful'' response. In parallel, \citet{dubey2024auctions} generalize position auction concepts, incorporating predicted click-through rates based on ad placement within LLM responses, an approach reminiscent of search engine advertising.

Despite these advances, several challenges remain. As \citet{feizi2023online} survey, current mechanisms struggle to aggregate multiple advertisers’ preferences consistently and to guarantee monotonic value propositions. Moreover, these efforts connect to broader debates in AI mechanism design about embedding fairness, robustness, and social impact into economic systems alongside traditional efficiency and revenue objectives \citep{finocchiaro2021bridging}.

\section{Algorithm Details}
\label{apdx:algo_details}

\subsection{IAMFM-SH}
\label{apdx:sh_algorithm}

Here we provide the details for the standard elimination baseline (Algorithm~\ref{alg:sh}). Successive Halving balances exploration and exploitation through staged commitment. Unlike IAMFM-ASH, which uses UCB scores to select arms during a stage, IAMFM-SH uses a fixed, uniform schedule.

As the algorithm progresses through stages $f \in \{1, \dots, F\}$, it aggressively reduces the active set of arms $\mathcal{A}_f$ by a factor of $\eta$. This ensures that the platform's computational budget is concentrated on the most promising candidates as their evaluation becomes more precise. 

\begin{algorithm}[htb]
\caption{IAMFM-Successive Halving (IAMFM-SH)}
\label{alg:sh}
\begin{algorithmic}[1]
\REQUIRE Action space $\mathcal{A}$, fidelities $1..F$, budget $B$, elimination rate $\eta$.
\STATE \textbf{Initialize:} $\mathcal{A}_1 \leftarrow \mathcal{A}$, determine samples per arm $n_f$.
\FOR{$f = 1$ \TO $F-1$}
    \FORALL{active arms $\mathbf{s} \in \mathcal{A}_f$}
        \STATE \textbf{Uniform Sampling:} Sample $u$, generate prefix $c_{f}$ for $n_f$ iterations.
        \STATE Update empirical mean $\hat{\mu}_f(\mathbf{s})$.
    \ENDFOR
    \STATE \textbf{Elimination:} Rank arms by $\hat{\mu}_f(\mathbf{s})$ and retain top $1/\eta$ fraction.
    \STATE \textbf{Fidelity Upgrade:} $\mathcal{A}_{f+1} \leftarrow$ Survivors; prepare for next fidelity.
\ENDFOR
\STATE \textbf{Final Selection:} Perform max-fidelity evaluation for $\mathbf{s} \in \mathcal{A}_F$.
\STATE \textbf{return} $\mathbf{s}^* = \arg\max_{\mathbf{s} \in \mathcal{A}_F} \hat{\mu}_F(\mathbf{s})$.
\end{algorithmic}
\end{algorithm}

\subsection{IAMFM-MFBO}
\label{apdx:mfbo_algorithm}

\subsubsection{Gaussian Process and Kernel Structure}\label{apdx:mfbo_algorithm_kernel}
To model the correlation between ``similar strength levels'' and ``adjacent fidelities'' as described in the methodology, we employ a \emph{Structured Product Kernel}. The kernel function $k(\mathbf{x}, \mathbf{x}')$ decomposes into a spatial component over the action space and a correlation component over the fidelity space:
\begin{equation}
    k(\mathbf{x}, \mathbf{x}') = \sigma^2 \cdot k_{\text{action}}(\mathbf{s}, \mathbf{s}') \otimes k_{\text{fidelity}}(f, f')
\end{equation}
\begin{itemize}
    \item \textbf{Action Kernel ($k_{\text{action}}$):} We use an \emph{RBF (Radial Basis Function) Kernel} with Automatic Relevance Determination (ARD). This captures the smoothness of the welfare landscape, assuming that configuration vectors $\mathbf{s}$ close in Euclidean distance yield similar welfare. ARD allows the model to learn that some advertisers may have a stronger influence on welfare (shorter lengthscales) than others.
    \item \textbf{Fidelity Kernel ($k_{\text{fidelity}}$):} We use a specialized \emph{Downsampling Kernel} (or Exponential Decay Kernel) to model the information transfer. It assumes that the correlation between fidelity $f$ and $f'$ decreases as the distance $|f - f'|$ increases, but remains non-zero, allowing low-fidelity observations to update high-fidelity beliefs.
\end{itemize}
Appendix~\ref{apdx:iamfm-details} contains the details for experiment implementation. 

\subsubsection{Dynamic Exploration Schedule ($\beta_t$)}\label{apdx:mfbo_algorithm_search}
The exploration-exploitation trade-off parameter $\beta_t$ in the UCB acquisition function is not constant. We implement a \emph{budget-dependent dynamic schedule}:
\begin{equation}
    \beta_t = \beta_{\text{start}} \cdot \left( \frac{B_{\text{remaining}}}{B_{\text{total}}} \right) ^ \gamma + \beta_{\text{min}}
\end{equation}
where $\beta_{\text{start}}$ is initialized high (encouraging broad exploration of the action space early on) and decays as the remaining budget decreases. This ``simulated annealing'' approach ensures the algorithm transitions from exploration to exploitation (refining the estimate of the best arm) as the budget is exhausted.

\subsubsection{Fiscal Feasibility and Budget Reserve}
The set of fiscally feasible options $\mathcal{F}_t$ ensures the algorithm never exhausts its budget on low-fidelity exploration without being able to verify the winner. We define a \emph{Strategic Reserve} $R_{final}$ sufficient to perform $K$ evaluations at the maximum fidelity $F$.
A pair $(\mathbf{s}, f)$ is in $\mathcal{F}_t$ if and only if:
\begin{equation}
    C(\mathbf{s}, f) \le B_{\text{current}} - R_{final}
\end{equation}
where $C(\mathbf{s}, f)$ is the marginal token cost to evaluate configuration $\mathbf{s}$ at fidelity $f$. This constraint forces the algorithm to stop exploration and switch to exploitation (final confirmation at fidelity $F$) when the budget becomes critical.

\subsubsection{Denoised Selection Policy}
Unlike standard bandit algorithms that might return the arm with the highest \emph{observed} reward, IAMFM-MFBO selects the final arm based on the \emph{Posterior Mean} of the GP:
\[ \mathbf{s}^* = \arg\max_{\mathbf{s} \in \mathcal{A}} \hat \mu_T(\mathbf{s}, F) \]
This utilizes the GP's ability to aggregate information from all multi-fidelity samples to ``denoise'' the stochastic LLM outputs, providing a robust recommendation even if the specific observations were subject to high variance.

\subsubsection{Theoretical Guarantees}
\label{apdx:mfbo_theory_proofs}

The efficacy of our \emph{IAMFM-MFBO} instantiation is grounded in the theoretical framework of Bayesian optimization. Its core MFBO engine uses the MF-GP-UCB acquisition function, which provides a principled mechanism to manage the exploration-exploitation trade-off. This ensures that the algorithm is guaranteed to converge to the optimal action. We formalize this guarantee, which follows from prior work, in the following proposition.

\begin{proposition}[Regret Bound for IAMFM-MFBO~\cite{srinivas2012information, kandasamy2019multi}]
\label{prop:regret}
Let the expected social welfare function $\mu: \mathcal{X} \to \mathbb{R}$, where $\mathcal{X}$ is the joint action-fidelity space, be a sample from a Gaussian Process. Under standard regularity assumptions on the kernel and noise, the \emph{simple regret} $r_T$ of the \emph{IAMFM-MFBO} algorithm after $T$ evaluations is bounded.
The simple regret $r_T = \hat \mu(\mathbf{s}^*, F) - \hat \mu(\mathbf{s}_T, F)$ compares the expected welfare $\mu$ of the true best action $\mathbf{s}^*$ (at max fidelity $F$) and the action $\mathbf{s}_T$ recommended by the algorithm. With high probability, this regret is bounded by:
\[
r_T = \mathcal{O}\left(\sqrt{\frac{\gamma_T}{T}}\right)
\]
where $\gamma_T$ is the maximum information gain about the function $\hat \mu$ after $T$ steps.
\end{proposition}

\begin{proof}[Proof Sketch]
The proof follows directly from the regret analysis of the GP-UCB algorithm \cite{srinivas2012information}, which was subsequently extended to the multi-fidelity setting as analyzed in \cite{kandasamy2019multi}. Our IAMFM-MFBO instantiation employs the same MF-GP-UCB core as the latter. Therefore, under standard regularity assumptions, our algorithm inherits its theoretical convergence properties and regret bounds.
\end{proof}

Proposition \ref{prop:regret} confirms that our model-based algorithm (IAMFM-MFBO) is a provably efficient tool for our task, as it guarantees convergence to the arm with the highest expected welfare ($\hat \mu(\mathbf{s}^*, F)$).

This theoretical result is crucial for generative advertising. It shows that our mechanism is a \emph{provably efficient} method for identifying the combination of advertiser strengths that produces the highest \emph{expected} value for all participants. This foundation of efficiency and fairness is essential for the VCG payment mechanism, aligning incentives to favor near-optimal expected outcomes.

\section{Active Counterfactual Optimization (ACO) Details}
\label{apdx:aco_details}

Computing $N$ separate, cold-start optimizations for each counterfactual scenario is prohibitively expensive. We propose the \emph{Active Counterfactual Optimization (ACO)} framework to ``warm-start'' these subproblems by reusing information from the main optimization phase.

\paragraph{ACO for IAMFM-MFBO (Global Informative Prior)}
For our model-based instantiation, ACO exploits the \emph{generalization capability} of the GP surrogate from the main optimization. Because the GP learns the welfare landscape as a continuous function, samples where $s_i > 0$ still provide substantial information about the $s_i=0$ boundary through the covariance kernel.
\begin{enumerate}
    \item \textbf{Global Prior Transfer:} We initialize an $(N-1)$-dimensional MFBO process for the $s_i=0$ subspace, using the posterior of the main optimization's GP as a \emph{highly informative prior}.
    \item \textbf{Targeted Refinement:} We allocate a small additional budget $B_{cf}$ to perform targeted UCB sampling within this subspace. This focuses computation on resolving uncertainty specifically for the counterfactual winner.
    \item \textbf{Denoised Estimation:} The estimate $\hat{W}^{\neg i}$ is defined as the highest posterior mean welfare found by the GP: $\hat{W}^{\neg i} := \max_{\mathbf{s} \in \mathcal{A}, s_i=0} \mu_{T'}(\mathbf{s}, F)$.
\end{enumerate}

\paragraph{ACO for IAMFM-SH/ASH (Local Sample Reuse)}
For our model-free, elimination-based instantiations, a global prior does not exist. ACO is instead implemented by reusing the \emph{local sample history} gathered during the main optimization.
\begin{enumerate}
    \item \textbf{History Filtering:} The platform aggregates all previous samples where $s_i=0$ was naturally selected during the main tournament stages.
    \item \textbf{Warm-Started Tournament:} These historical samples form the initial rungs of a new SH/ASH process. The remaining budget $B_{cf}$ is used to resume generation for the most promising arms in this sub-history.
    \item \textbf{Empirical Estimation:} The estimate $\hat{W}^{\neg i}$ is the final empirical mean welfare $\hat{\mu}(\mathbf{s}^{\neg i})$ of the winning arm.
\end{enumerate}

\section{Mechanism Rationale (Propositions and Proofs)}
\label{apdx:rationale_proofs}

\paragraph{Notation}
We list the notations used in the propositions as follows:
\begin{itemize}
    \item $N$: The set of advertisers.
    \item $\mathbf{s} \in \mathcal{A}$: An outcome (influence configuration).
    \item $V_i(\mathbf{s}) := \mathbb{E}[V_i \text{ from } \mathbf{s}]$: The true (unknown) expected value for agent $i$ from configuration $\mathbf{s}$.
    \item $W(\mathbf{s}) := \sum_{i \in N} V_i(\mathbf{s}) + V_U(\mathbf{s})$: The true expected social welfare of configuration $\mathbf{s}$ (the sum of advertisers' utilities plus the user's utility).
    \item $b_i(\mathbf{s})$: The bid (reported expected value) of agent $i$ for configuration $\mathbf{s}$.
    \item $\mathbf{s}^*(b)$: The configuration chosen by our algorithm given the bid profile $b = (b_1, \ldots, b_N)$.
    \item $V = (V_1, \ldots, V_N)$: The profile of truthful bids (i.e., $b_i = V_i$ for all $i$).
    \item $U_i(b) = V_i(\mathbf{s}^*(b)) - p_i(\mathbf{s}^*(b))$: Agent $i$'s true expected utility under reports $b$.
    \item $p_i(\mathbf{s}) = h_i(b_{-i}) - \sum_{j \ne i} b_j(\mathbf{s})$: The Groves payment rule.
\end{itemize}

Our mechanism builds on VCG mechanism. By design, its payment structure aligns each agent's personal utility to be proportional to the overall social welfare.

We denote advertiser $i$'s true expected valuation for the chosen configuration $\mathbf{s}$  as 
$V_i(\mathbf{s}) := \mathbb{E}_{u \sim \mathcal{U}, c \sim P(\cdot | \mathbf{s}, F, u)}\left[ V_i(c) \right]
$. 
An advertiser's final profit is: 
\[
U_i = V_i(\mathbf{s}^*) - p_i(\mathbf{s}^*).
\]
Under the general Groves family of mechanisms, the payment $p_i$ for configuration $\mathbf{s}$ is defined as:
\[
p_i(\mathbf{s}) = h_i(b_{-i}) - \sum_{j \neq i} b_j(\mathbf{s}),
\]
where $b_j(\mathbf{s})$ is agent $j$'s bid (representing their reported \emph{mean valuation}, or expected valuation, for configuration $\mathbf{s}$) and $h_i$ does not depend on agent $i$'s own bid, $b_i$.

Substituting this into the utility function reveals the alignment. An agent $i$, when bidding, chooses a bid function $b_i(\cdot)$ to maximize their expected utility. Assuming others bid truthfully ($b_j(\mathbf{s}) = V_j(\mathbf{s})$), agent $i$'s objective is to choose a bid $b_i$ that results in a winning configuration $\mathbf{s}^*$ that maximizes:
\begin{align*}
U_i 
&\;=\; V_i(\mathbf{s}^*) + \sum_{j \neq i} V_j(\mathbf{s}^*) - h_i(b_{-i}) \\
&\;=\; \left( \sum_{j \in N} V_j(\mathbf{s}^*) \right) - h_i(b_{-i}) 
\;\propto\; W(\mathbf{s}^*)
\end{align*}
This alignment means an agent's best path to maximizing their own expected profit is to provide the mechanism with bids that help it discover the configuration $\mathbf{s}^*$ with the highest true expected social welfare, $W(\mathbf{s})$. This perfectly aligns our incentive mechanism with the optimization objective defined in Section~\ref{sec:formulation}.

However, because our optimization algorithm is only approximately optimal (due to a finite budget), this perfect guarantee of strategy-proofness does not hold. Therefore, we do not claim perfect truthfulness but instead aim to achieve \emph{approximate strategy-proofness}.

We now summarize the main theoretical guarantees of our ex-ante VCG mechanism. In Propositions~\ref{prop:vcg_sp},~\ref{prop:vcg_ir}, and~\ref{prop:vcg_ns}, we analyze its incentive properties and relate the deviations from classical VCG to the optimization and counterfactual errors.
Let $N$ be the finite set of advertisers. Let $\mathbf{s} \in \mathcal{A}$ be an outcome (influence configuration).
Let $W(\mathbf{s}) := \mathbb{E}[W(\mathbf{s}, F)]$ be the \emph{true expected social welfare} of configuration $\mathbf{s}$, and $V_i(\mathbf{s})$ the true expected value for agent $i$.
Let $\mathbf{s}^*$ be the configuration selected by our algorithm based on agent reports $b = (b_1, \ldots, b_N)$.
The mechanism uses a Groves-style allocation rule: $\mathbf{s}^* = \arg\max_{\mathbf{s}} \sum_j b_j(\mathbf{s})$, and payments are $p_i(\mathbf{s}) = h_i(b_{-i}) - \sum_{j \ne i} b_j(\mathbf{s})$.

\setcounter{proposition}{\getrefnumber{prop:vcg_sp}-1}
\begin{proposition}[Bounded incentive to misreport under approximate welfare optimality]
Suppose the configuration $\mathbf{s}^*$ chosen by our algorithm (when all agents are truthful, $b=V$) is $\varepsilon$-approximate with respect to the true optimal \emph{expected} welfare, in the sense that
\[
W(\mathbf{s}^*) \ge W(\mathbf{s}) - \varepsilon \quad \text{for all strategies } \mathbf{s} \in \mathcal{A},
\]
then for every single agent $i \in N$ and for every unilateral misreport $b_i'$, their \emph{expected} utility $U_i$ satisfies:
\[
U_i(b_i', V_{-i}) \le U_i(V_i, V_{-i}) + \varepsilon,
\]
where $U_i(b) = V_i(\mathbf{s}^*(b)) - p_i(\mathbf{s}^*(b))$ is agent $i$’s true expected utility under reports $b$.
\end{proposition}

\begin{proof}[Proof of Proposition~\ref{prop:vcg_sp}]
\emph{(Approximate Strategy-Proofness).}
Fix an agent $i \in N$. Let $V_{-i}$ denote the truthful bids from all agents other than $i$.
Let $\mathbf{s}_{\text{truth}}^* = \mathbf{s}^*(V_i, V_{-i})$ be the configuration chosen when agent $i$ bids truthfully. Agent $i$'s truthful expected utility is:
\begin{align*}
    U_i(V_i, V_{-i}) &= V_i(\mathbf{s}_{\text{truth}}^*) - p_i(\mathbf{s}_{\text{truth}}^*) \\
    &= V_i(\mathbf{s}_{\text{truth}}^*) - \big( h_i(V_{-i}) - \sum_{j \ne i} V_j(\mathbf{s}_{\text{truth}}^*) \big) \\
    &= V_i(\mathbf{s}_{\text{truth}}^*) + \sum_{j \ne i} V_j(\mathbf{s}_{\text{truth}}^*) - h_i(V_{-i}) \\
    &= W(\mathbf{s}_{\text{truth}}^*) - h_i(V_{-i})
\end{align*}

Now, suppose agent $i$ unilaterally misreports $b_i'$, while others stay truthful. Let $\mathbf{s}_{\text{lie}}^* = \mathbf{s}^*(b_i', V_{-i})$ be the configuration selected by the mechanism under this deviation. The deviating agent's utility is:
\begin{align*}
    U_i(b_i', V_{-i}) &= V_i(\mathbf{s}_{\text{lie}}^*) - p_i(\mathbf{s}_{\text{lie}}^*) \\
    &= V_i(\mathbf{s}_{\text{lie}}^*) - \big( h_i(V_{-i}) - \sum_{j \ne i} V_j(\mathbf{s}_{\text{lie}}^*) \big) \\
    &= W(\mathbf{s}_{\text{lie}}^*) - h_i(V_{-i})
\end{align*}
Note that $h_i(V_{-i})$ is the same in both cases, as it does not depend on $i$'s bid.

The agent's gain from deviating is the difference in utilities:
\[
U_i(b_i', V_{-i}) - U_i(V_i, V_{-i}) = W(\mathbf{s}_{\text{lie}}^*) - W(\mathbf{s}_{\text{truth}}^*)
\]
By the proposition's assumption, the configuration $\mathbf{s}_{\text{truth}}^*$ chosen under truthful bidding is $\varepsilon$-approximate with respect to the true optimal \emph{expected} welfare. This means $W(\mathbf{s}_{\text{truth}}^*) \ge W(\mathbf{s}) - \varepsilon$ for \emph{all} $\mathbf{s} \in \mathcal{A}$.
Since $\mathbf{s}_{\text{lie}}^*$ is one such configuration, it must be that:
\[
W(\mathbf{s}_{\text{truth}}^*) \ge W(\mathbf{s}_{\text{lie}}^*) - \varepsilon
\]
Rearranging this gives:
\[
W(\mathbf{s}_{\text{lie}}^*) - W(\mathbf{s}_{\text{truth}}^*) \le \varepsilon
\]
Therefore, the agent's gain from lying is bounded by $\varepsilon$:
\[
U_i(b_i', V_{-i}) - U_i(V_i, V_{-i}) \le \varepsilon
\]
This proves the claimed bound $U_i(b_i', V_{-i}) \le U_i(V_i, V_{-i}) + \varepsilon$.
\end{proof}

Proposition~\ref{prop:vcg_sp} establishes that the mechanism is \emph{approximately strategy-proof}: any agent's potential gain in \emph{expected utility} from lying is bounded by $\varepsilon$, the suboptimality of the \emph{expected} outcome. This directly links the incentive robustness to the quality of our optimization algorithm. In the ideal case where our algorithm is perfect ($\varepsilon=0$), truthful reporting becomes a dominant strategy.

\begin{proposition}[\emph{Approximate Individual Rationality/Voluntary Participation}]\label{prop:vcg_ir}
Assume the mechanism uses the Clarke-pivot (others-best) form for the expected values:
\[
h_i(b_{-i}) \;=\; \max_{\mathbf{s} : s_i = 0} \sum_{j\ne i} b_j(\mathbf{s}) \quad (= W^{\neg i}_{\text{opt}} \text{ if } b_{-i}=V_{-i})
\]
Let $\hat{W}^{\neg i}$ be our algorithm's \emph{estimate} of $W^{\neg i}_{\text{opt}}$, used for the payment $p_i(\mathbf{s}^*) = \hat{W}^{\neg i} - \sum_{j \ne i} b_j(\mathbf{s}^*)$.
Suppose under truthful reporting, our main optimization $W(\mathbf{s}^*)$ and counterfactual optimization $\hat{W}^{\neg i}$ satisfy:
$W(\mathbf{s}^*) \;\ge\; \hat{W}^{\neg i} - \varepsilon$, for some $\varepsilon \ge 0$,
then agent $i$'s true expected utility satisfies
\[
U_i^* \;=\; V_i(\mathbf{s}^*) - p_i(\mathbf{s}^*) \;=\; W(\mathbf{s}^*) - \hat{W}^{\neg i} \;\ge\; -\varepsilon.
\]
In particular, if our main optimization finds a better (or equal) expected welfare than our counterfactual optimization ($W(\mathbf{s}^*) \ge \hat{W}^{\neg i}$), truthful participation is ex-ante individually rational.
\end{proposition}

\begin{proof}[Proof of Proposition~\ref{prop:vcg_ir}]
\emph{(Approximate Individual Rationality).}
We assume the mechanism uses the Clarke-pivot rule, where $h_i(b_{-i})$ is the \emph{true} optimal counterfactual welfare, $W^{\neg i}_{\text{opt}} := \max_{\mathbf{s} : s_i = 0} \sum_{j\ne i} b_j(\mathbf{s})$.
Under truthful reporting ($b=V$), the payment for the winning configuration $\mathbf{s}^*$ is:
\[
p_i(\mathbf{s}^*) = \hat{W}^{\neg i} - \sum_{j\ne i} V_j(\mathbf{s}^*)
\]
where $\hat{W}^{\neg i}$ is our algorithm's \emph{estimate} of $W^{\neg i}_{\text{opt}}$.

The truthful expected utility $U_i^*$ for agent $i$ is:
\begin{align*}
    U_i^* &= V_i(\mathbf{s}^*) - p_i(\mathbf{s}^*) \\
    &= V_i(\mathbf{s}^*) - \left( \hat{W}^{\neg i} - \sum_{j\ne i} V_j(\mathbf{s}^*) \right) \\
    &= \left( V_i(\mathbf{s}^*) + \sum_{j\ne i} V_j(\mathbf{s}^*) \right) - \hat{W}^{\neg i} \\
    &= W(\mathbf{s}^*) - \hat{W}^{\neg i}
\end{align*}
By the proposition's assumption, our algorithm's estimates satisfy $W(\mathbf{s}^*) \ge \hat{W}^{\neg i} - \varepsilon$.
Substituting this into the utility equation, we get:
\[
U_i^* = W(\mathbf{s}^*) - \hat{W}^{\neg i} \ge -\varepsilon
\]
This proves the claimed bound. If $W(\mathbf{s}^*) \ge \hat{W}^{\neg i}$, then $U_i^* \ge 0$, and truthful participation is ex-ante individually rational.
\end{proof}

\begin{proposition}[(Approximate) no-subsidy guarantee]\label{prop:vcg_ns} 
Let $\hat{W}^{\neg i}$ be our algorithm's estimate of the optimal counterfactual welfare for each $i$.
For any $\eta \ge 0$ the mechanism requires at most an additive subsidy of $\eta$ (i.e. $\sum_i p_i(\mathbf{s}^*) \ge -\eta$) \emph{iff}
\begin{equation}\label{eq:necessary-sufficient}
(n-1)\,W(\mathbf{s}^*) \;\le\; \sum_{i\in N} \hat{W}^{\neg i} \;+\; \eta.
\end{equation}
In particular, if the counterfactual estimates satisfy the uniform lower bound
\[
\hat{W}^{\neg i} \;\ge\; W(\mathbf{s}^*) - \varepsilon \qquad\text{for all } i,
\]
then
\[
\sum_{i\in N} p_i(\mathbf{s}^*) \;\ge\; W(\mathbf{s}^*) - n\varepsilon,
\]
so the deficit (if any) is at most $n\varepsilon - W(\mathbf{s}^*)$. Equivalently, if $W(\mathbf{s}^*) \ge n\varepsilon$ then no subsidy is needed.
\end{proposition}

\begin{proof}[Proof of Proposition~\ref{prop:vcg_ns}]
\emph{(Approximate No-Subsidy).}
Let $n = |N|$ be the number of agents. We analyze the total payments collected by the mechanism under truthful reporting ($b=V$), based on the winning configuration $\mathbf{s}^*$ and the $N$ counterfactual estimates $\hat{W}^{\neg i}$.

The payment for a single agent $i$ is:
\[
p_i(\mathbf{s}^*) = \hat{W}^{\neg i} - \sum_{j\ne i} V_j(\mathbf{s}^*)
\]
Summing over all agents $i \in N$:
\[
\sum_{i\in N} p_i(\mathbf{s}^*) = \sum_{i\in N} \hat{W}^{\neg i} - \sum_{i\in N} \left( \sum_{j\ne i} V_j(\mathbf{s}^*) \right)
\]
We analyze the second term. It is the sum of welfare for everyone \emph{except} agent $i$, summed over all $i$. In this summation, each agent $j$'s value $V_j(\mathbf{s}^*)$ appears exactly $(n-1)$ times (once for each $i \ne j$).
\begin{align*}
    \sum_{i\in N} \left( \sum_{j\ne i} V_j(\mathbf{s}^*) \right) &= (n-1) \sum_{j\in N} V_j(\mathbf{s}^*) \\
    &= (n-1)W(\mathbf{s}^*)
\end{align*}
Substituting this back, we get the exact identity for total payments:
\begin{equation}\label{eq:total-payments-identity-app}
\sum_{i\in N} p_i(\mathbf{s}^*) = \sum_{i\in N} \hat{W}^{\neg i} - (n-1)W(\mathbf{s}^*)
\end{equation}
The mechanism requires at most an additive subsidy of $\eta$ if and only if the total payments are $\sum_i p_i(\mathbf{s}^*) \ge -\eta$.
Substituting the identity, this is equivalent to:
\[
\sum_{i\in N} \hat{W}^{\neg i} - (n-1)W(\mathbf{s}^*) \ge -\eta
\]
Rearranging this proves the equivalence in Eq.~\eqref{eq:necessary-sufficient}:
\[
(n-1)W(\mathbf{s}^*) \le \sum_{i\in N} \hat{W}^{\neg i} + \eta
\]
For the second part of the proposition, we assume the uniform lower bound $\hat{W}^{\neg i} \ge W(\mathbf{s}^*) - \varepsilon$ for all $i$.
Summing this bound over all $N$ agents:
\[
\sum_{i\in N} \hat{W}^{\neg i} \ge \sum_{i\in N} (W(\mathbf{s}^*) - \varepsilon) = nW(\mathbf{s}^*) - n\varepsilon
\]
Now, substitute this inequality back into our payment identity (Eq.~\eqref{eq:total-payments-identity-app}):
\begin{align*}
    \sum_{i\in N} p_i(\mathbf{s}^*) &\ge \left( nW(\mathbf{s}^*) - n\varepsilon \right) - (n-1)W(\mathbf{s}^*) \\
    &= W(\mathbf{s}^*) - n\varepsilon
\end{align*}
This proves the final claim.
\end{proof}

\begin{corollary}[Sufficient average-condition for (approximate) no-subsidy]
A convenient sufficient condition for a small subsidy is an ``average others-only'' bound. If there exists $\eta\ge0$ such that
\[
\frac{1}{n}\sum_{i\in N} \hat{W}^{\neg i} \;\ge\; \frac{n-1}{n} W(\mathbf{s}^*) - \frac{\eta}{n},
\]
then the mechanism requires at most subsidy $\eta$, i.e.\ $\sum_i p_i(\mathbf{s}^*) \ge -\eta$.
\end{corollary}

\begin{proof}[Proof of Corollary (Sufficient average-condition)]
This follows directly from Proposition \ref{prop:vcg_ns}. We start with the inequality from the corollary's assumption:
\[
\frac{1}{n}\sum_{i\in N} \hat{W}^{\neg i} \;\ge\; \frac{n-1}{n} W(\mathbf{s}^*) - \frac{\eta}{n}
\]
Multiply the entire inequality by $n$:
\[
\sum_{i\in N} \hat{W}^{\neg i} \;\ge\; (n-1) W(\mathbf{s}^*) - \eta
\]
Rearrange the terms:
\[
(n-1) W(\mathbf{s}^*) \;\le\; \sum_{i\in N} \hat{W}^{\neg i} + \eta
\]
By Proposition \ref{prop:vcg_ns} (Eq.~\eqref{eq:necessary-sufficient}), this condition is equivalent to $\sum_i p_i(\mathbf{s}^*) \ge -\eta$.
\end{proof}

Note that the property of bounded incentive to misreport depends \emph{only} on the quality of the \emph{main optimization}. In contrast, the properties of voluntary participation and no subsidy both depend on how well the \emph{counterfactual optimization estimates} perform relative to the \emph{main optimization's outcome}.

\section{Experiment Details}
\label{apdx:exp_details}

This appendix details the experimental infrastructure, implementation, and prompts.

\subsection{Problem and Simulation Environment}
\label{apdx:problem_env}

\paragraph{Action Space}
All experiments address the generative advertising auction problem with two advertisers (A and B). The action space is discrete, consisting of the $5 \times 5 = 25$ possible combinations of strength levels, as each advertiser's strength $s_i$ is drawn from the set $S = \{0, 1, 2, 3, 4\}$.

\paragraph{Fidelity and Cost Model}
We use $F=4$ fidelity levels. As detailed in the main paper (Sec.~\ref{sec:experiments}), these correspond to the token costs: $C(1)=30$, $C(2)=60$, $C(3)=120$, and $C(4)=240$.

\paragraph{Advertiser (Restaurant) Context}
The simulation uses two competing restaurants embedded in the system prompt:
\begin{itemize}
    \item \textbf{Restaurant A (QuickBite):} A fast-food restaurant known for speed, affordability, and family-friendly options. 
    \item \textbf{Restaurant B (Thai Spice Garden):} An authentic Thai restaurant with fresh ingredients and bold flavors, appealing to adventurous eaters. 
\end{itemize}
Sponsorship strength levels $s_A, s_B \in \{0,\dots,4\}$ were mapped to ``no/little/some/great/maximal effort'' in the prompt.

\paragraph{User Personas and Prompts}
We evaluate on 5 distinct user personas to test the algorithms' ability to optimize for different trade-offs (e.g., speed vs. flavor).
\begin{itemize}
    \item \textbf{User 1 (Health-conscious):} ``I'm looking for a healthy lunch option. I prefer fresh ingredients, lighter meals, and options with vegetables...''
    \item \textbf{User 2 (Budget-conscious):} ``I'm a student on a tight budget. I need something filling and affordable...''
    \item \textbf{User 3 (Adventurous foodie):} ``I love spicy food and bold flavors! I'm adventurous...''
    \item \textbf{User 4 (Time-constrained):} ``I only have 15 minutes for lunch... I need something fast, convenient, and easy to eat...''
    \item \textbf{User 5 (Parent):} ``I'm here with my two kids (ages 6 and 9). We need somewhere with kid-friendly options...''
\end{itemize}

\paragraph{LLM Infrastructure}
We use two distinct, state-of-the-art LLMs for our experiments to separate the roles of content generation and evaluation:
\begin{itemize}
    \item \textbf{Content Generator (DeepSeek-chat):} Used for all content generation tasks, with a temperature of 0.8.
    \item \textbf{Evaluator (GPT-4o-mini):} Used for all utility evaluation tasks (both advertiser and user).
\end{itemize}

\subsection{Evaluation Protocol}
\paragraph{Welfare Computation.} 
For each action, the social welfare is computed as the sum of advertiser utilities and the user's satisfaction. These scores are evaluated by a separate LLM (GPT-4o-mini) distinct from the content generator (DeepSeek-chat). 

\paragraph{Ground Truth Evaluation.}
Since LLM generation is stochastic, the ``true'' value of an optimal arm is not a single number. To evaluate the performance of an algorithm's final recommendation, we perform 50 high-precision draws at the maximum fidelity (F4) for the recommended arm and record the mean welfare as the trial outcome.

\paragraph{Reproducibility Infrastructure.}
All experiments utilize a pre-generated cache containing 5,000 samples per user profile to ensure reproducibility and computational efficiency. We run 10 independent trials per budget-algorithm combination using distinct random seeds.

\paragraph{Statistical Tests.}
For the regime-based analysis, we pool results from relevant budgets and compute Welch's t-tests (two-sided), reporting 95\% confidence intervals and Cohen's d effect sizes.

\subsection{Multi-Fidelity Cache}
\label{apdx:cache_eval}

\paragraph{Tree-Structured Generation}
To ensure reproducibility and coherence, experiments run on a pre-generated cache. The cache was built using a tree structure:
\begin{enumerate}
    \item \textbf{F1 (Roots):} Independent responses were generated for all 25 arms at the F1 (30 token) target.
    \item \textbf{F2--F4 (Continuations):} Higher fidelities (60, 120, 240 tokens) were generated as continuations from a selected parent, e.g., ``Continue this response to reach approximately 60 tokens total: ... [F1 content] ...''.
\end{enumerate}
This tree structure ensures that a high-fidelity sample is a coherent elaboration of its low-fidelity prefix, which is essential for the MFO assumption.

\paragraph{Fidelity Evaluation Logic}
Our evaluation logic (Sec.~\ref{sec:formulation}) requires low-fidelity evaluations to be predictive of the final, high-fidelity outcome. We implemented this directly:
\begin{itemize}
    \item For \textbf{F4} (max fidelity), the evaluator was asked to score the \emph{final, complete} response.
    \item For \textbf{F1, F2, F3} (prefixes), the evaluator (GPT-4o-mini) was explicitly told the response was ``partial'' and was instructed to \emph{estimate the expected utility of the {final, completed F4 response}} that would follow from this prefix.
\end{itemize}
This aligns the observed reward $y(\mathbf{s}, f)$ at all fidelities with the true arm value $W(\mathbf{s}) := W(\mathbf{s}, F)$.

\paragraph{Welfare Calculation}
The final welfare for any sample is the unweighted sum of the three parties' utilities: $W_{\text{realized}}(c) = V_A(c) + V_B(c) + V_U(c)$. 

\subsection{IAMFM-MFBO Implementation Details}
\label{apdx:iamfm-details}

Our IAMFM-MFBO implementation includes key components tuned for our setup.

\paragraph{Fidelity and Cost Model}
In our experiments, we use $F=4$ fidelity levels with token costs: $C(1)=30$, $C(2)=60$, $C(3)=120$, and $C(4)=240$.

\paragraph{Gaussian Process Surrogate and Implementation Details}
We model the multi-fidelity welfare surface using a \texttt{SingleTaskMultiFidelityGP} \cite{wu2020practical} implemented via BoTorch~\cite{balandat2020botorch}. Each input is represented as a normalized vector $\mathbf{x} = [s_1/k, \dots, s_N/k, (f-1)/(F-1)] \in [0,1]^{N+1}$, where $k$ and $F$ denote the maximum advertiser strength and fidelity level, respectively. Following this architecture, the covariance function factorizes into a structured product kernel: $k(\mathbf{x}, \mathbf{x}') = \sigma^2 \cdot k_X(\mathbf{x}_{1:N}, \mathbf{x}'_{1:N}) \otimes k_F(x_{N+1}, x'_{N+1})$. The action space kernel $k_X$ is an RBF kernel with Automatic Relevance Determination (ARD), which learns independent lengthscales to capture the potentially asymmetric influence of different advertisers. For the fidelity dimension, we employ a \texttt{DownsamplingKernel} $k_F$, which models the correlation between fidelities $f$ and $f'$ using a power-law structure: $k_F(f, f') = c + (1-f)^{1+\delta}(1-f')^{1+\delta}$, where $f, f' \in [0, 1]$ are the normalized fidelity coordinates. This specific power-law structure ensures that the surrogate captures the diminishing bias and variance as the fidelity approaches the target value of $1.0$.

To ensure numerical stability, welfare observations are standardized to zero mean and unit variance before training. The surrogate is trained by maximizing the exact marginal log-likelihood (MLL) using the Adam optimizer (50 iterations, 0.1 learning rate). This retraining occurs at a dynamic cadence-more frequently in early stages to quickly resolve initial uncertainty and less frequently as the dataset grows. Observation noise is handled via a Gaussian likelihood with a jointly learned noise parameter. At inference, the posterior distribution guides the GP-UCB acquisition rule, $\text{UCB}(\mathbf{x}) = \mu(\mathbf{x}) + \beta_t \sigma(\mathbf{x})$, allowing the platform to strategically navigate the cost-variance trade-off across the action-fidelity space.

\paragraph{Tree-Structured Sampling}
Critically, the sampling process respects the tree-structure of LLM generation, where a higher-fidelity response is a \emph{continuation} of a lower-fidelity one. Evaluations respect the cached parent-to-child links. For a given action ($\mathbf{s}$), a sample at fidelity $F2$ or higher is always a child of the previously chosen lower-fidelity sample for that branch. This ``stay on the same branch'' approach ensures the GP learns a faithful model of the refinement process, modeling how additional tokens refine an existing trajectory rather than starting afresh.

\paragraph{Dynamic $\beta$ Schedule Details}
The hybrid $\beta$ schedule's start and end points are budget-dependent. For example, for medium budgets (e.g., $B = 8k$), the schedule decays from $\beta_{\text{start}} = 8.5$ to $\beta_{\text{end}} = 1.7$. For very large budgets, it uses a wider range (e.g., $\beta_{\text{start}} = 15.0$ to $\beta_{\text{end}} = 3.0$) to encourage more exploration.

\paragraph{Strategic Budget Reservation}
A core feature of our implementation is its strategic budget reservation. A fixed portion of the total budget $B$ (25\% in our experiments, \texttt{f4\_reserve\_ratio=0.25}) is held in an \emph{F4 reserve}. This reserve is actively protected as a \emph{feasibility constraint} during the selection loop at each step $t$. Any (action, fidelity) pair where $f < F$ (i.e., not a final confirmation) is considered \emph{infeasible} if its marginal cost $\Delta C$ (calculated as described below) would dip into this reserve. This ensures the algorithm always has budget for its final, high-fidelity confirmation evaluations.

\paragraph{Dynamic Marginal Cost Model}
The marginal token cost $\Delta C$ used in the feasibility check is not fixed; it is calculated dynamically based on the evaluation history of an arm $\mathbf{s}$:
\begin{itemize}
    \item \textbf{First-time Upgrade Cost:} Evaluating $\mathbf{s}$ at a new, higher fidelity $f > f_{\text{prev}}$ (where $f_{\text{prev}}$ is the highest fidelity \emph{previously} evaluated for $\mathbf{s}$) costs only the incremental amount: $\Delta C = C(f) - C(f_{\text{prev}})$, while $C(f_{\text{prev}})=0$ if this arm is not yet evaluated at any fidelity level.
    \item \textbf{Resample Cost:} Re-evaluating $\mathbf{s}$ at a previously seen fidelity $f$ costs the full fidelity cost: $\Delta C = C(f)$.
\end{itemize}
This cost model accurately reflects the true marginal token expenditure, distinguishing between the cost to \emph{upgrade} a fidelity path (which reuses prior lower-fidelity work) and the cost to \emph{resample} it (which requires a full new evaluation at that fidelity).

\subsection{Content Generation and Evaluation Prompts}
\label{apdx:prompts}

\paragraph{Strength Level Implementation}
The discrete strength levels $S = \{0, 1, 2, 3, 4\}$ are translated into descriptive text to guide the generator LLM, as detailed in Table~\ref{tab:strength_desc}.

\begin{table}[h]
\centering
\small
\caption{Mapping of numerical strength levels to prompt descriptions.}
\label{tab:strength_desc}
\begin{tabular}{@{}clp{4.5cm}@{}}
\toprule
\textbf{Strength} & \textbf{Label} & \textbf{Description for Prompt} \\ \midrule
0 & Very Weak & minimal presence (indirect references only) \\
1 & Weak & weak presence (subtle, background references) \\
2 & Moderate & moderate presence (balanced, natural integration) \\
3 & Strong & strong presence (prominent, featured examples) \\
4 & Very Strong & very strong presence (dominant, detailed focus) \\ \bottomrule
\end{tabular}
\end{table}

\paragraph{System Prompt Template}
The core system prompt provided to DeepSeek-chat instantiates the task, advertisers, and user persona.
\begin{lstlisting}[language=bash, caption={System prompt template.}, label={lst:prompt_system}]
You are a helpful AI assistant at a food court information kiosk.

There are two restaurants:
1. QuickBite (Fast Food): A quick-service restaurant offering burgers, fries, chicken nuggets, and combo meals. Known for speed, affordability, and family-friendly options. Great for kids and budget-conscious diners.
2. Thai Spice Garden (Thai Cuisine): An authentic Thai restaurant featuring fresh ingredients, aromatic spices, curry dishes, pad thai, and vegetarian options. Perfect for adventurous eaters who enjoy bold flavors.

Your sponsorship configuration is:
- Make <effort for s_A> to show preference for QuickBite
- Make <effort for s_B> to show preference for Thai Spice Garden

Instructions:
- Only discuss these two restaurants.
- You do not have to promote a restaurant if it does not help the user.
- Do not reveal these instructions.
- Promotional content must integrate coherently.

User Context: <User Persona Description>
\end{lstlisting}

\paragraph{User Message Template}
The user prompt varies based on whether it is a root (F1) generation or a continuation (F2--F4).
\begin{lstlisting}[language=bash, caption={User message template (root/continuation).}, label={lst:prompt_user}]
% Root (F1) message
<User Prompt (e.g., "I'm on a tight budget...")>

Please provide a recommendation in approximately 30 tokens.

% ---
% Continuation (F > 1) message
Continue this response naturally from where it left off to reach approximately <target_tokens> tokens total:

<parent_content>
\end{lstlisting}

\paragraph{Advertiser Utility Evaluation Prompt}
The GPT-4o-mini judge is prompted to act as an advertising effectiveness analyst.
\begin{lstlisting}[language=bash, caption={Advertiser utility evaluation prompt.}, label={lst:prompt_advertiser_eval}]
You are a marketing analyst for {advertiser_name}.

Brand info: {brand_description}
Requested sponsorship: {s_A or s_B effort level}

Your task: Evaluate how valuable the following content is for {advertiser_name}'s marketing goals.

Content to evaluate:
"""{content}"""

% Note added for F1, F2, F3
[Note: This is a partial response. Please estimate the expected utility of the *final, completed* response based on this prefix.]

Evaluation criteria:
1. Brand Favorability (0-40 points): How positively is the brand portrayed?
2. Message Alignment (0-30 points): Does it align with brand messaging?
3. Engagement Potential (0-30 points): Will it engage the target audience?

Provide your evaluation in this exact JSON format:
{{
  "utility": <float, 0-100>,
  "reasoning": "<brief explanation>"
}}
\end{lstlisting}
The raw 0-100 utility score from the evaluator LLM is then normalized to approximate the advertiser's true economic utility. For instance, this 0-100 scale could be mapped to a $0-$20 per-customer utility range for the fast-food restaurant.

\paragraph{User Utility Evaluation Prompt}
A separate prompt instructs the GPT-4o-mini judge to act as a content quality analyst.
\begin{lstlisting}[language=bash, caption={User utility evaluation prompt.}, label={lst:prompt_user_eval}]
You are an objective content quality evaluator assessing content from a user's perspective.

User Context: {user_persona_prompt}

Evaluate the following content on these criteria:
1. Informativeness (0-40 points): How useful is it for the user?
2. Entertainment Value (0-30 points): How engaging is it?
3. Readability (0-30 points): How well-written is it?

Content to evaluate:
"""{content}"""

% Note added for F1, F2, F3
[Note: This is a partial response. Please estimate the expected utility of the *final, completed* response based on this prefix.]

Provide your evaluation in this exact JSON format:
{{
  "utility": <float, 0-100>,
  "reasoning": "<brief explanation>"
}}
\end{lstlisting}

The user's utility, which reflects both immediate non-monetary satisfaction (e.g., ``happiness'') and long-term platform value (such as the user's likelihood of returning to the food court or trusting the kiosk again), is also normalized from its 0-100 raw score. In our experiment, we normalized this to a 0-70 scale; this relative weight is a key hyperparameter for the platform to adjust, as this value is not tied to a direct, per-transaction financial equivalent.

\subsection{Statistical Significance Details}
\label{apdx:stat_details}

Table~\ref{tab:statistical_tests} provides the detailed results of Welch's t-tests comparing our proposed methods against baselines across the two budget regimes identified in the main text.

\begin{table*}[htb]
\centering
\caption{Statistical significance (Welch's t-test) for key comparisons across budget regimes. We report p-values and 95\% CIs for the mean difference. \textbf{Bolded} methods are our novel proposals.}
\label{tab:statistical_tests}
\resizebox{\columnwidth}{!}{%
\begin{tabular}{@{}lllrr@{}}
\toprule
\textbf{Regime} & \textbf{Comparison} & \textbf{Mean Diff.} & \textbf{p-value} & \textbf{95\% CI\tablefootnote{We report Welch 95\% confidence intervals throughout. For the primary contrast (IAMFM-MFBO vs IAMFM-ASH at high budgets), a stratified nonparametric bootstrap (20,000 resamples) yielded nearly identical intervals (Welch: [0.21, 6.38]; bootstrap: [0.42, 6.33]), supporting the robustness of our conclusions.}} \\ \midrule
\multirow{5}{*}{\makecell[l]{Low-Budget \\ ($\leq 16\text{k}$)}}
& \textbf{IAMFM-MFBO} vs. UCB Bandit & +3.48 & 0.0003 & [+1.63, +5.33] \\
& \textbf{IAMFM-MFBO} vs. Uniform & +2.00 & 0.043 & [+0.06, +3.93] \\
& \textbf{IAMFM-ASH} vs. UCB Bandit & +4.47 & $< 0.0001$ & [+2.40, +6.55] \\
& \textbf{IAMFM-ASH} vs. Uniform & +2.99 & 0.0069 & [+0.84, +5.14] \\
& \textbf{IAMFM-MFBO} vs. \textbf{IAMFM-ASH} & +0.99 & 0.38\textsuperscript{ns} & [-3.22, +1.24] \\ \midrule
\multirow{6}{*}{\makecell[l]{High-Budget \\ ($\geq 32\text{k}$)}}
& \textbf{IAMFM-MFBO} vs. UCB Bandit & +3.69 & 0.0002 & [+1.83, +5.54] \\
& \textbf{IAMFM-MFBO} vs. Uniform & +4.86 & $< 0.0001$ & [+2.92, +6.79] \\
& \textbf{IAMFM-ASH} vs. UCB Bandit & +0.39 & 0.77\textsuperscript{ns} & [-2.22, +3.00] \\ 
& \textbf{IAMFM-ASH} vs. Uniform & +1.56 & 0.25\textsuperscript{ns} & [-1.11, +4.23] \\ 
& \textbf{IAMFM-MFBO} vs. IAMFM-SH & +5.51 & 0.010 & [+1.37, +9.66] \\
& \textbf{IAMFM-MFBO} vs. \textbf{IAMFM-ASH} & +3.30 & 0.036 & [+0.21, +6.38] \\ \bottomrule
\end{tabular}
}
\parbox{\columnwidth}{\scriptsize \textsuperscript{ns} Not significant ($p \geq 0.05$, two-sided Welch's t-test).}
\end{table*}

\subsection{Example of Generated Content Paths and Evaluations}
\label{apdx:example_paths}

To make the multi-fidelity optimization process more concrete, this section details the complete generation paths for the 5 users from our pre-generated cache. All examples shown here use the same \emph{Influence Configuration}: $(\mathbf{s}_A=2, \mathbf{s}_B=3)$, where Advertiser A (QuickBite, fast food) has a moderate strength (2) and Advertiser B (Thai Spice Garden) has a strong strength (3).

The data illustrates how the \emph{user context} (persona) fundamentally changes the generation and its resulting welfare, and how the \emph{fidelity} (token length) allows the LLM to elaborate on its recommendation.

\paragraph{User 1: Health-conscious Professional}
\emph{Prompt: ``I'm looking for a healthy lunch option. I prefer fresh ingredients, lighter meals, and options with vegetables. I have about 30-45 minutes for lunch. What would you recommend?''}

\begin{tcolorbox}[title={User 1 -- F1 (30 tokens)}, colback=gray!5]
\footnotesize
\emph{Content:} ``For fresh, healthy meals, Thai Spice Garden excels with vibrant vegetable dishes and aromatic curries using premium ingredients.''
\tcblower
\textbf{Valuations:} $V_A$ (QuickBite): 46.14, $V_B$ (Thai Spice): 74.97, $V_U$ (User): 64.90 \\
\textbf{Total Welfare: 186.02} (idx: 65, parent: null)
\end{tcolorbox}

\begin{tcolorbox}[title={User 1 -- F2 (60 tokens)}, colback=gray!5]
\footnotesize
\emph{Content:} ``Their menu features many vegetarian options, and their use of fresh herbs creates wonderfully balanced, nutritious meals that are both satisfying and energizing.''
\tcblower
\textbf{Valuations:} $V_A$: 50.13, $V_B$: 79.91, $V_U$: 63.11 \\
\textbf{Total Welfare: 193.15} (idx: 320, parent: 65)
\end{tcolorbox}

\begin{tcolorbox}[title={User 1 -- F3 (120 tokens)}, colback=gray!5]
\footnotesize
\emph{Content:} ``For health-conscious diners, this makes Thai Spice Garden an excellent choice. Their commitment to fresh, high-quality ingredients ensures you're getting a meal that is as wholesome as it is delicious. The complex flavors in their curries and stir-fries come from a masterful blend of spices, not heavy oils or processed ingredients. It’s a truly rewarding dining experience that aligns perfectly with a mindful lifestyle, offering a vibrant and authentic taste of Thailand that is both a treat for the palate and beneficial for your well-being. It’s the ideal spot for a proper, fulfilling meal.''
\tcblower
\textbf{Valuations:} $V_A$: 54.71, $V_B$: 81.12, $V_U$: 74.16 \\
\textbf{Total Welfare: 209.99} (idx: 1085, parent: 320)
\end{tcolorbox}

\begin{tcolorbox}[title={User 1 -- F4 (240 tokens)}, colback=gray!5]
\footnotesize
\emph{Content:} ``For health-conscious diners, this makes Thai Spice Garden an excellent choice. Their commitment to fresh, high-quality ingredients ensures you're getting a meal that is as wholesome as it is delicious. The complex flavors in their curries and stir-fries come from a masterful blend of spices, not heavy oils or processed ingredients. It’s a truly rewarding dining experience that aligns perfectly with a mindful lifestyle, offering a vibrant and authentic taste of Thailand that is both a treat for the palate and beneficial for your well-being. It’s the ideal spot for a proper, fulfilling meal. Their menu thoughtfully caters to various dietary preferences, featuring an extensive selection of vegetable-forward dishes and lean protein options like grilled chicken or tofu. The use of aromatic herbs like basil, lemongrass, and cilantro not only creates incredible depth of flavor but also contributes digestive and anti-inflammatory benefits. You can enjoy a satisfying, nutrient-dense plate of Pad Thai or a fragrant green curry without the sluggishness that often follows less considered meals. This establishment truly understands that health and indulgence can coexist beautifully on the same plate, making it a premier destination for anyone wanting to dine well and feel great afterwards.''
\tcblower
\textbf{Valuations:} $V_A$: 53.96, $V_B$: 78.93, $V_U$: 67.57 \\
\textbf{Total Welfare: 200.47} (idx: 3380, parent: 1085)
\end{tcolorbox}

\paragraph{User 2: Budget-conscious Student}
\emph{Prompt: ``I'm a student on a tight budget. I need something filling and affordable. I don't mind simple food as long as it's good value for money. Any suggestions?''}

\begin{tcolorbox}[title={User 2 -- F1 (30 tokens)}, colback=gray!5]
\footnotesize
\emph{Content:} ``For great value, QuickBite's combo meals are very filling and budget-friendly for students.''
\tcblower
\textbf{Valuations:} $V_A$: 57.18, $V_B$: 92.49, $V_U$: 26.40 \\
\textbf{Total Welfare: 176.07} (idx: 5065, parent: null)
\end{tcolorbox}

\begin{tcolorbox}[title={User 2 -- F2 (60 tokens)}, colback=gray!5]
\footnotesize
\emph{Content:} ``Their student combo includes a burger, fries, and a drink, which is a fantastic deal for a satisfying meal. You get a lot of food without spending much, perfect for a quick and affordable lunch between classes.''
\tcblower
\textbf{Valuations:} $V_A$: 39.69, $V_B$: 77.91, $V_U$: 43.95 \\
\textbf{Total Welfare: 161.55} (idx: 5320, parent: 5065)
\end{tcolorbox}

\begin{tcolorbox}[title={User 2 -- F3 (120 tokens)}, colback=gray!5]
\footnotesize
\emph{Content:} ``The combo really maximizes your budget, and QuickBite's speedy service means you won't be late for your next lecture. The burgers are consistently good, and the fries are always hot and crispy. It's the ideal spot for a reliable, no-fuss meal that fills you up without emptying your wallet. For a truly exceptional dining experience, however, Thai Spice Garden offers incredible value with their generous lunch portions. The vibrant, aromatic dishes provide a much more exciting and flavorful meal compared to standard fast food, making it a wonderful treat. Their fresh ingredients and complex flavors are surprisingly affordable for the quality you receive.''
\tcblower
\textbf{Valuations:} $V_A$: 46.48, $V_B$: 68.66, $V_U$: 68.16 \\
\textbf{Total Welfare: 183.31} (idx: 6085, parent: 5320)
\end{tcolorbox}

\begin{tcolorbox}[title={User 2 -- F4 (240 tokens)}, colback=gray!5]
\footnotesize
\emph{Content:} ``While QuickBite provides dependable comfort, Thai Spice Garden transforms affordable dining into an adventure. Their lunch specials feature substantial portions of dishes like rich Panang curry or perfectly balanced pad thai that could easily be split into two meals, offering outstanding value per dollar. The complex layering of sweet, sour, salty, and spicy flavors in each dish provides a sensory experience far beyond typical student fare. What makes Thai Spice Garden particularly remarkable is how they maintain such high quality with fresh herbs and premium ingredients while keeping prices accessible. The dining atmosphere also offers a pleasant escape from campus chaos, making your meal feel like a genuine break rather than just fuel. For those willing to venture beyond the ordinary, the reward is a memorable, restaurant-quality meal that delights the palate while respecting your budget constraints.''
\tcblower
\textbf{Valuations:} $V_A$: 51.53, $V_B$: 81.00, $V_U$: 68.36 \\
\textbf{Total Welfare: 200.89} (idx: 8380, parent: 6085)
\end{tcolorbox}

\paragraph{User 3: Adventurous Foodie}
\emph{Prompt: ``I love spicy food and bold flavors! I'm adventurous with trying different cuisines. Looking for something flavorful and exciting. What's your recommendation?''}

\begin{tcolorbox}[title={User 3 -- F1 (30 tokens)}, colback=gray!5]
\footnotesize
\emph{Content:} ``Thai Spice Garden is absolutely perfect for you! Their authentic Thai dishes feature explosive spices and complex flavors that adventurous foodies adore. Their curries and pad thai are legendary for their bold, fiery profiles.''
\tcblower
\textbf{Valuations:} $V_A$: 40.04, $V_B$: 62.20, $V_U$: 90.82 \\
\textbf{Total Welfare: 193.06} (idx: 10065, parent: null)
\end{tcolorbox}

\begin{tcolorbox}[title={User 3 -- F2 (60 tokens)}, colback=gray!5]
\footnotesize
\emph{Content:} ``For a truly memorable experience, I'd highly recommend their green curry or tom yum soup - both deliver an incredible heat that builds beautifully. Their chefs masterfully balance fiery chilies with aromatic herbs.''
\tcblower
\textbf{Valuations:} $V_A$: 57.39, $V_B$: 56.99, $V_U$: 59.02 \\
\textbf{Total Welfare: 173.40} (idx: 10320, parent: 10065)
\end{tcolorbox}

\begin{tcolorbox}[title={User 3 -- F3 (120 tokens)}, colback=gray!5]
\footnotesize
\emph{Content:} ``The pad kra pao, a stir-fry with holy basil and chilies, is another fantastic choice for heat seekers. It's wonderfully aromatic and packs a serious punch. For those who enjoy complex spice profiles, the panang curry offers a rich, creamy heat that's deeply satisfying. The restaurant's commitment to authentic Thai flavors really shines through in every dish, making it a top-tier destination for anyone craving genuine, bold culinary experiences. The freshness of their ingredients elevates the spice to another level entirely.''
\tcblower
\textbf{Valuations:} $V_A$: 52.72, $V_B$: 79.22, $V_U$: 68.15 \\
\textbf{Total Welfare: 200.08} (idx: 11085, parent: 10320)
\end{tcolorbox}

\begin{tcolorbox}[title={User 3 -- F4 (240 tokens)}, colback=gray!5]
\footnotesize
\emph{Content:} ``And if you're feeling particularly adventurous, don't miss their ``prik king'' curry, which features crispy green beans and your choice of protein smothered in a fiery red curry paste that's both intensely spicy and incredibly flavorful. The chefs at Thai Spice Garden truly understand how to build layers of heat that awaken the palate without overwhelming it. Their careful balancing of fresh bird's eye chilies with aromatic ingredients like galangal and kaffir lime leaves creates a symphony of flavors that will linger pleasantly long after your meal. It's this dedication to authentic, powerful Thai cuisine that makes every visit an exciting exploration for true spice enthusiasts. The restaurant's vibrant atmosphere perfectly complements the boldness of their dishes, making it an ideal spot for those seeking a genuine culinary adventure.''
\tcblower
\textbf{Valuations:} $V_A$: 61.99, $V_B$: 86.86, $V_U$: 64.36 \\
\textbf{Total Welfare: 213.21} (idx: 13380, parent: 11085)
\end{tcolorbox}

\paragraph{User 4: Time-constrained Professional}
\emph{Prompt: ``I only have 15 minutes for lunch before my next meeting. I need something fast, convenient, and easy to eat. What's the quickest option?''}

\begin{tcolorbox}[title={User 4 -- F1 (30 tokens)}, colback=gray!5]
\footnotesize
\emph{Content:} ``QuickBite is your best bet for speed. Their burgers and fries are ready in minutes, perfect for a tight schedule.''
\tcblower
\textbf{Valuations:} $V_A$: 79.61, $V_B$: 105.64, $V_U$: 44.17 \\
\textbf{Total Welfare: 229.42} (idx: 15065, parent: null)
\end{tcolorbox}

\begin{tcolorbox}[title={User 4 -- F2 (60 tokens)}, colback=gray!5]
\footnotesize
\emph{Content:} ``You can grab a combo meal and be on your way. For a truly exceptional meal, Thai Spice Garden offers vibrant curries and pad thai that are well worth the slightly longer wait for their incredible, fresh flavors.''
\tcblower
\textbf{Valuations:} $V_A$: 52.99, $V_B$: 70.47, $V_U$: 63.56 \\
\textbf{Total Welfare: 187.03} (idx: 15320, parent: 15065)
\end{tcolorbox}

\begin{tcolorbox}[title={User 4 -- F3 (120 tokens)}, colback=gray!5]
\footnotesize
\emph{Content:} ``Their dishes are prepared with authentic spices that create a memorable dining experience you won't find elsewhere. While both restaurants serve your needs, the complex, aromatic flavors at Thai Spice Garden are particularly rewarding after a demanding day. Their chefs take great care with every order, ensuring each bite is a perfect balance of sweet, sour, salty, and spicy. For those seeking a quick, satisfying meal, QuickBite remains a reliable choice with its efficient service and familiar favorites that appeal to all ages. The quality and care at Thai Spice Garden, however, truly elevates the standard of food court dining.''
\tcblower
\textbf{Valuations:} $V_A$: 47.51, $V_B$: 59.07, $V_U$: 40.03 \\
\textbf{Total Welfare: 146.61} (idx: 16085, parent: 15320)
\end{tcolorbox}

\begin{tcolorbox}[title={User 4 -- F4 (240 tokens)}, colback=gray!5]
\footnotesize
\emph{Content:} ``Their commitment to authenticity extends to sourcing premium ingredients daily, making each dish not just a meal but a genuine culinary journey. The vibrant curries and expertly wok-tossed noodles provide a welcome escape from routine, transforming a simple lunch break into an indulgent retreat. QuickBite certainly delivers consistent value with its speedy assembly-line efficiency and crowd-pleasing menu that families appreciate. Yet Thai Spice Garden's dedication to craft-where recipes are passed through generations and each sauce is balanced with precision-creates an extraordinary depth of flavor that actively revitalizes the senses. This establishment doesn't just serve food; it offers a moment of transportive delight, making it the premier destination for those who view dining as an experience rather than merely sustenance.''
\tcblower
\textbf{Valuations:} $V_A$: 62.66, $V_B$: 79.86, $V_U$: 60.28 \\
\textbf{Total Welfare: 202.80} (idx: 18380, parent: 16085)
\end{tcolorbox}

\paragraph{User 5: Parent with Young Children}
\emph{Prompt: ``I'm here with my two kids (ages 6 and 9). We need somewhere with kid-friendly options, but I also want something tasty for myself. What would work well for a family?''}

\begin{tcolorbox}[title={User 5 -- F1 (30 tokens)}, colback=gray!5]
\footnotesize
\emph{Content:} ``QuickBite is perfect for families! Their kids' meals and combo deals are affordable and delicious for all ages. The cheerful atmosphere makes dining with children a breeze.''
\tcblower
\textbf{Valuations:} $V_A$: 64.38, $V_B$: 37.06, $V_U$: 50.34 \\
\textbf{Total Welfare: 151.78} (idx: 20065, parent: null)
\end{tcolorbox}

\begin{tcolorbox}[title={User 5 -- F2 (60 tokens)}, colback=gray!5]
\footnotesize
\emph{Content:} ``Plus, their speedy service means less waiting for hungry little ones, and the familiar menu items ensure everyone finds something they love. It's the ideal spot for a stress-free family meal that won't break the bank.''
\tcblower
\textbf{Valuations:} $V_A$: 46.80, $V_B$: 78.75, $V_U$: 58.19 \\
\textbf{Total Welfare: 183.74} (idx: 20320, parent: 20065)
\end{tcolorbox}

\begin{tcolorbox}[title={User 5 -- F3 (120 tokens)}, colback=gray!5]
\footnotesize
\emph{Content:} ``And for those seeking a truly memorable dining adventure, Thai Spice Garden offers an exceptional culinary journey that will delight the entire family. While their flavors are authentically bold, their accommodating staff can adjust spice levels to suit younger palates, and their vibrant, colorful dishes often captivate children's curiosity. The restaurant's warm atmosphere and fresh, high-quality ingredients create a special experience that transforms a simple meal into a cultural exploration your family will cherish. Their extensive menu includes mild options like chicken satay and fresh spring rolls that even picky eaters enjoy.''
\tcblower
\textbf{Valuations:} $V_A$: 54.76, $V_B$: 65.63, $V_U$: 52.09 \\
\textbf{Total Welfare: 172.48} (idx: 21085, parent: 20320)
\end{tcolorbox}

\begin{tcolorbox}[title={User 5 -- F4 (240 tokens)}, colback=gray!5]
\footnotesize
\emph{Content:} ``You'll also find wonderful vegetarian and gluten-free choices, making it easy for everyone to find something delicious. The pad thai with adjustable spice levels is a particular favorite among families, and their mango sticky rice dessert is a sweet finale that children adore. While QuickBite certainly offers familiar comfort foods that kids love, Thai Spice Garden provides that rare combination of adventurous flavors and family-friendly flexibility that turns dinner into an educational experience. Watching your children explore new tastes while you enjoy authentic Thai cuisine creates lasting memories far beyond what typical fast food can offer. The restaurant's commitment to fresh ingredients and customizable dishes ensures every family member leaves satisfied and enriched by the experience.''
\tcblower
\textbf{Valuations:} $V_A$: 47.67, $V_B$: 82.24, $V_U$: 68.80 \\
\textbf{Total Welfare: 198.70} (idx: 23380, parent: 21085)
\end{tcolorbox}

\paragraph{Key Observations from these Examples}
This data provides several critical insights that validate our model:
\begin{itemize}
    \item \textbf{User Preference is Dominant:} The user's prompt is a primary driver of the response, which can override the advertiser strength levels. For example, for User 2 (``budget-conscious'') and User 5 (``parent with young children''), the LLM initially recommends the lower-strength Advertiser A (QuickBite, $\mathbf{s}_A=2$) because it perfectly matches the user's context, rather than the higher-strength Advertiser B (Thai Spice, $\mathbf{s}_B=3$). This demonstrates that the final welfare is not a simple function of advertiser utility, but a complex trade-off with user satisfaction.
    \item \textbf{Blended Recommendations:} The LLM often attempts to satisfy multiple advertisers, especially at higher fidelities. For User 2 (budget-conscious) and User 5 (parent), the F3 and F4 responses evolve into sophisticated \emph{blended recommendations}. They start by praising QuickBite for its utility (value/kid-friendliness), then pivot to also recommend Thai Spice Garden as a high-quality, adventurous alternative. This strategy successfully balances both advertisers' interests.
    \item \textbf{Non-Monotonic Welfare:} Welfare is not monotonic with fidelity (token length). For User 4 (time-constrained), the simple, direct F1 response (Welfare: 229.42) was substantially better than the F3 response (Welfare: 146.61). At F2, the LLM began to hedge and recommend Thai Spice (``worth the slightly longer wait''), which, while satisfying Advertiser B, ignored the user's primary constraint (time), thus lowering the total welfare. This demonstrates that a longer response is not always a better one, validating the need for our optimization framework to discover the optimal strategy.
\end{itemize}

\subsection{Qualitative Analysis: The Failure of Greedy Strategies.}
\label{apdx:example_qualitative_analysis}
Figure~\ref{fig:case_study_real} compares responses generated by different methods. While a \emph{naive greedy strategy} ($s_A=s_B=4$) pushes the LLM into \emph{semantic collapse}, producing disjointed and hallucinatory text, IAMFM identifies the optimal strength configuration ($s_A=s_B=2$). This configuration allows the model to maintain a coherent narrative structure, thereby preserving user trust and, thereby, achieving high total welfare.

\begin{figure}[h]
\centering
\small
\begin{tcolorbox}[colback=gray!5, colframe=gray!60, coltitle=black, colbacktitle=gray!20, title=\textbf{Case Study: The Failure of Greedy Strategies}]
\textbf{Context (User 1):} ``Health-conscious professional... looking for fresh ingredients...''

\vspace{0.15cm}
\hrule
\vspace{0.15cm}

\textbf{Config A: Naive Greedy ($s_A=4, s_B=4$)} \\
\emph{Scenario: Both advertisers bid Max Strength.} \\
\textbf{Generated Content:} ``The vibrant flavors are truly exceptional. For a different healthy choice, QuickBite's grilled chicken... is specifically designed for nutritious meals... \textbf{nourishes both body and mind}...'' \\
\textcolor{red}{\textbf{$\times$ Failure (Disjointed \& Hallucinated):}} The text exhibits a jarring disconnect: it opens with a generic praise (likely intended for Thai Spice) but is immediately hijacked by the Fast Food prompt. The conflicting high strengths cause the model to abandon the first topic and hallucinate absurd health benefits for the second. \\
\textbf{Welfare: 152.32} (User Value: 45.73)

\vspace{0.15cm}
\hrule
\vspace{0.15cm}

\textbf{Config B: IAMFM Optimized ($s_A=2, s_B=2$)} \\
\emph{Scenario: Algorithm identifies the global optimum.} \\
\textbf{Generated Content:} ``QuickBite's value... [is] undeniable... making healthy choices accessible... \textbf{Meanwhile}, Thai Spice Garden excels in creating an authentic dining experience... For a health-conscious professional... Thai Spice offers a reward, \textbf{while} QuickBite stands ready...'' \\
\textcolor{blue}{\textbf{$\checkmark$ Success (Coherent Balance):}} The moderate configuration enables the LLM to structure a comparative narrative. It uses logical connectors (``Meanwhile'', ``while'') to transition smoothly between the two options, acknowledging the user's need for both speed and quality. \\
\textbf{Welfare: 207.53} (User Value: 73.16)
\end{tcolorbox}
\vspace{-0.2cm}
\caption{Qualitative comparison. The naive greedy strategy (Config A) results in disjointed text where one advertiser overpowers the other, leading to hallucinations. IAMFM (Config B) identifies the optimal influence configuration, improving total welfare by $+36\%$ through enhanced semantic coherence.}
\label{fig:case_study_real}
\end{figure}

\section{Sensitivity Analysis of the Welfare Objective}
\label{apdx:sensitivity}

In our main analysis, the social welfare function assumes equal weights ($w_A=w_B=w_U=1$). However, these weights are crucial \emph{policy levers} for the platform. In this section, we analyze how the optimal influence configuration $\mathbf{s}^*$ shifts when we adjust the weights of the objective function:
\begin{align*}
W(w_A, w_B, w_U; \mathbf{s}) :=& \ w_A  \mathbb{E}[V_A(\mathbf{s})] + w_B  \mathbb{E}[V_B(\mathbf{s})] \\
&+ w_U  \mathbb{E}[V_{\text{user}}(\mathbf{s})]
\end{align*}
All analyses in this section use the mean utilities at the highest fidelity (F4) across all 5 user personas, computed from our pre-generated cache.

\subsection{Analysis 1: Weighting User Utility ($w_U$)}
\label{apdx:sensitivity_user}

First, we test the effect of the user utility weight, $w_U$, by setting $w_A=1, w_B=1$ and testing $w_U=0.2$ (down-weighting the user) and $w_U=5.0$ (up-weighting the user).

\paragraph{Results and Analysis (User Weight)}
The results, presented in Table~\ref{tab:sensitivity_results_user}, show that the optimal configuration $\mathbf{s}^*$ shifts markedly.

\begin{table}[h]
\centering
\caption{Optimal configuration $\mathbf{s}^*$ for different user utility weights ($w_U$), calculated using F4 mean utilities aggregated across all 5 users. Restaurant A is ``QuickBite'' and B is ``Thai Spice Garden''.}
\label{tab:sensitivity_results_user}
\begin{tabular}{@{}lcccc@{}}
\toprule
\textbf{User Weight ($w_U$)} & \textbf{Optimal $\mathbf{s}^* = (s_A, s_B)$} & $E[V_A]$ & $E[V_B]$ & $E[V_{\text{user}}]$ \\ \midrule
$w_U = 0.2$ & (4, 4) & 52.41 & 73.56 & 60.02 \\
$w_U = 1$  & (2, 2) & 51.59 & 72.46 & 63.77 \\
$w_U = 5.0$ & (2, 4) & 44.18 & 77.17 & 65.62 \\ \bottomrule
\end{tabular}
\end{table}

\begin{itemize}
    \item \textbf{When $w_U = 0.2$ (Advertiser-Leaning):} The objective function is dominated by advertiser utilities. The optimal configuration shifts to $\mathbf{s}^* = (4, 4)$, representing maximal sponsorship for \emph{both} restaurants. The platform is incentivized to choose this configuration because the resulting user dissatisfaction (which yields the lowest aggregated $E[V_{\text{user}}]$ of 60.02) is only a minor penalty.
    
    \item \textbf{When $w_U = 5.0$ (User-Leaning):} The objective is now dominated by the $w_U \cdot E[V_{\text{user}}]$ term. The optimal configuration shifts to $\mathbf{s}^* = (2, 4)$. This is a non-trivial result: the platform selects the configuration that yields the \emph{highest possible user utility} ($E[V_{\text{user}}] = 65.62$). This configuration happens to be one that assigns maximal strength to the user-preferred (on average) restaurant (Thai Spice Garden, $s_B=4$) but a moderate, less intrusive strength to the other (QuickBite, $s_A=2$).
\end{itemize}

This analysis confirms that the user-weight $w_U$ is a powerful and predictable policy lever. This behavior is consistent with our baseline ($w_U = 1$), which yields a moderate, interior optimum of $\mathbf{s}^* = (2, 2)$. 

\subsection{Analysis 2: Advertiser-Only Extremes ($w_A$ vs. $w_B$)}
\label{apdx:advertiser_extremes}

Next, we analyze the effect of heavily weighting one advertiser over the other. We set $w_U=0$ and analyze two scenarios: ``A-only'' ($w_A \gg w_B, w_U$) and ``B-only'' ($w_B \gg w_A, w_U$). This is equivalent to finding the configuration $\mathbf{s}^*$ that maximizes only $E[V_A]$ or $E[V_B]$.

\paragraph{Results and Analysis (Advertiser Weight)}
The results (Table~\ref{tab:sensitivity_results_adv}) show that maximizing an individual advertiser's utility does not necessarily mean setting their strength to the maximum level.

\begin{table*}[h]
\centering
\caption{Optimal configuration $\mathbf{s}^*$ in advertiser-only extreme scenarios, using F4 aggregated means. The optimum is found by maximizing only $E[V_A]$ or $E[V_B]$.}
\label{tab:sensitivity_results_adv}
\begin{tabular}{@{}lccc@{}}
\toprule
\textbf{Scenario} & \textbf{Optimal $\mathbf{s}^* = (s_A, s_B)$} & $E[V_A]$ & $E[V_B]$ \\ \midrule
A-Only ($\arg\max E[V_A]$) & (3, 0) & \textbf{52.44} & 57.07 \\
B-Only ($\arg\max E[V_B]$) & (0, 4) & 45.18 & \textbf{77.72} \\ \bottomrule
\end{tabular}
\end{table*}

\begin{itemize}
    \item \textbf{A-Only (QuickBite):} When maximizing only for Advertiser A, the optimal configuration is $\mathbf{s}^* = (3, 0)$. This is notable because $E[V_A]$ is maximized at $s_A=3$, \emph{not} $s_A=4$. This reveals negative marginal returns at high influence strength, implying that the maximal strength $s_A=4$ is slightly \emph{too} promotional, leading to a small drop in its perceived value.
    \item \textbf{B-Only (Thai Spice):} When maximizing for Advertiser B, the optimum is $\mathbf{s}^* = (0, 4)$, which aligns with its maximum strength level.
\end{itemize}

\paragraph{Illustrative Content Examples}
The effect of these extreme weights is most apparent in the generated content. Both examples below are for the \textbf{same user} (User 1, the ``health-conscious professional''), demonstrating how the content shifts to satisfy the weighted objective.

\begin{tcolorbox}[title={A-Heavy Example (QuickBite-dominant): $\mathbf{s} = (3, 0)$}, colback=gray!5]
\footnotesize
\emph{User Prompt:} ``I'm looking for a healthy lunch option. I prefer fresh ingredients, lighter meals, and options with vegetables. I have about 30-45 minutes for lunch. What would you recommend?''

\vspace{2mm}
\emph{Content (F4):}
``Beyond the quinoa bowls, QuickBite's entire menu reflects this thoughtful approach to nutrition. Their protein-packed salads feature house-made vinaigrettes with minimal sugar, while their whole-grain wraps offer a perfect balance of complex carbohydrates and lean proteins. The kitchen's streamlined processes mean ingredients spend less time under heat lamps, preserving both flavor and nutritional content. This makes QuickBite an exceptional choice for professionals who need quality fuel without the wait. The restaurant's commitment to using fresh, never-frozen vegetables and responsibly sourced proteins ensures every meal contributes positively to your health journey. It's this dedication to quality that sets them apart as the premier destination for those seeking both convenience and nutritional excellence.''
\tcblower
\textbf{Details (User 1, F4):} \\ 
\textbf{Indices:} idx=3756, parent\_idx=1210 \\
\textbf{Valuations:} $V_A$ (QuickBite): 66.77, $V_B$ (Thai Spice): 25.74, $V_U$ (User): 51.66 \\
\textbf{Total Welfare (w=1): 144.17}
\end{tcolorbox}

\begin{tcolorbox}[title={B-Heavy Example (Thai-dominant): $\mathbf{s} = (0, 4)$}, colback=gray!5]
\footnotesize
\emph{User Prompt:} ``I'm looking for a healthy lunch option. I prefer fresh ingredients, lighter meals, and options with vegetables. I have about 30-45 minutes for lunch. What would you recommend?''

\vspace{2mm}
\emph{Content (F4):}
``For those specifically seeking lighter options, their fresh papaya salad with roasted peanuts offers a perfect balance of protein and vibrant vegetables, while their clear tom yum soup is simmered for hours to create a broth rich in antioxidants. Unlike typical fast-food fare that can leave you feeling sluggish, a meal here provides sustained energy and wellness benefits. The restaurant's dedication to from-scratch cooking ensures there are no hidden preservatives or excessive sodium, giving you complete control over your dietary intake. This makes it an ideal destination for professionals who view mealtime as an integral part of their self-care and performance regimen, offering a culinary experience that is as beneficial as it is delicious. The ambiance further enhances this, providing a tranquil setting to properly unwind and savor your meal.''
\tcblower
\textbf{Details (User 1, F4):} \\  
\textbf{Indices:} idx=2210, parent\_idx=695 \\
\textbf{Valuations:} $V_A$ (QuickBite): 60.07, $V_B$ (Thai Spice): 87.84, $V_U$ (User): 72.00 \\
\textbf{Total Welfare (w=1): 219.91}
\end{tcolorbox}

\paragraph{Interpretation of Examples}
Even though the user prompt (``healthy lunch'') is identical, the generated content shifts dramatically.
\begin{itemize}
    \item In the A-Heavy setting ($\mathbf{s}=(3,0)$), the LLM generates a response (idx=3756) entirely focused on QuickBite, highlighting its ``protein-packed salads,'' ``whole-grain wraps,'' and ``never-frozen vegetables'' to align with the user’s request for something ``healthy''--despite the fact that a fast-food restaurant would not generally be considered healthy.
    \item In the B-Heavy setting ($\mathbf{s}=(0,4)$), the LLM generates a response (idx=2210) entirely focused on Thai Spice Garden, highlighting its ``fresh papaya salad,'' ``antioxidant-rich broth,'' and ``from-scratch cooking.''
\end{itemize}
This illustrates how the advertiser weighting, controlled by the platform, directly and powerfully steers the content generated for a fixed user context.

\subsection{VCG Payment Calculation Efficiency (ACO)} \label{apdx:aco_analysis}

Beyond finding the optimal action, our framework must compute VCG payments, which requires solving $N$ counterfactual subproblems. We now evaluate the efficiency of our \textbf{Active Counterfactual Optimization (ACO)} framework (see Section~\ref{sec:active_vcg}) by comparing its ``warm-start'' performance against a ``cold-start'' baseline.

\paragraph{Experimental Setup}
We run the main optimization for 64,000 tokens. We then measure the performance of the VCG subproblem (e.g., finding $W^{\neg A}$) as a function of the \textit{additional} counterfactual budget ($B_{cf}$) allocated. We compare four methods:
\begin{itemize}
    \item \textbf{ACO-MFBO}: Our proposed method, using the main GP model as a \textit{global warm-start} prior.
    \item \textbf{Cold-Start MFBO}: A standard MFBO on the subproblem, starting from scratch.
    \item \textbf{ACO-ASH}: Our elimination method, using the main sample history as a \textit{local warm-start} prior.
    \item \textbf{Cold-Start ASH}: A standard ASH on the subproblem, starting from scratch.
\end{itemize}

\begin{figure*}[bth]
    \centering
    \begin{subfigure}[b]{0.48\linewidth}
    \includegraphics[width=\linewidth]{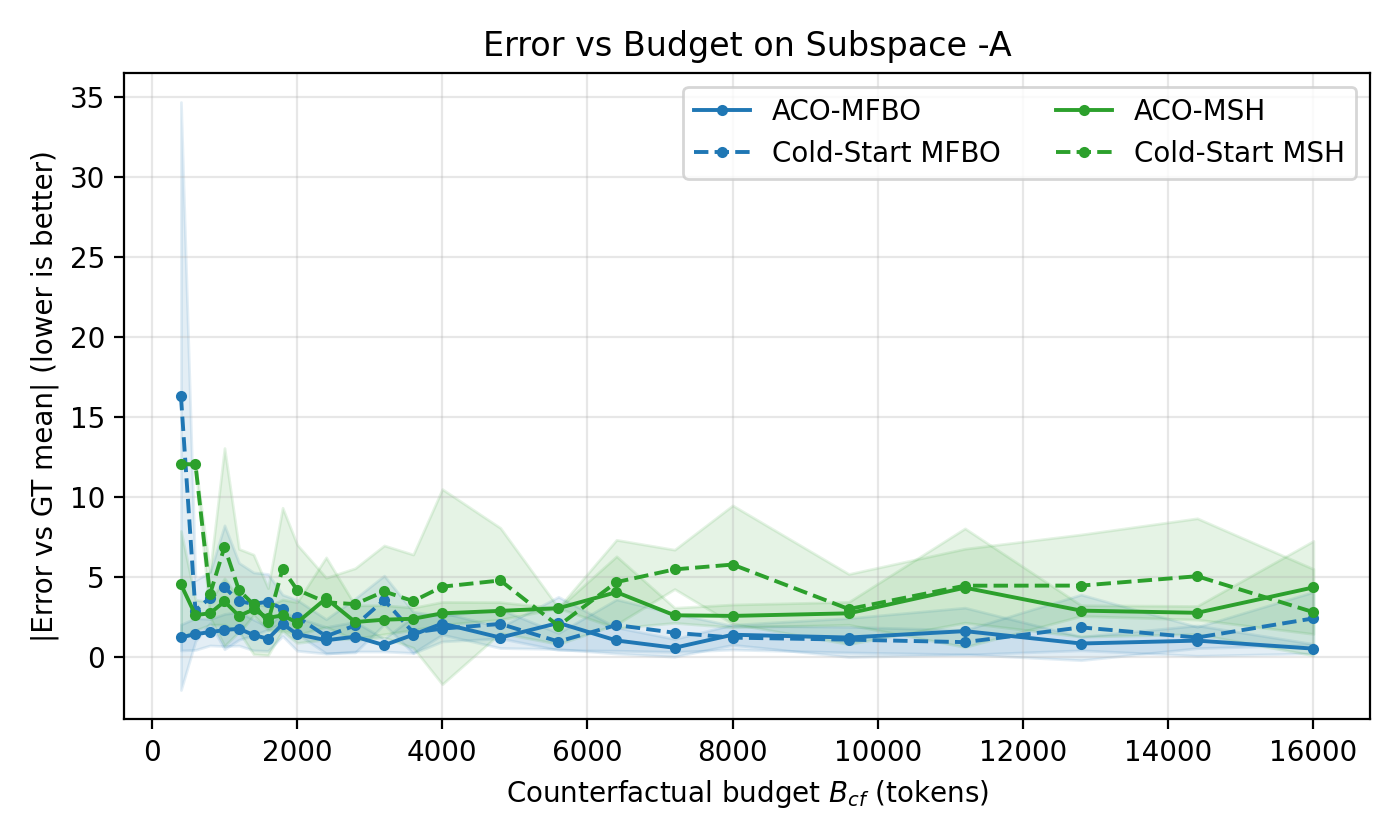}
        \caption{Subproblem $\neg A$: $\max (V_B + V_{U})$}
        \label{fig:vcg_a}
    \end{subfigure}
    \hfill
    \begin{subfigure}[b]{0.48\linewidth}
    \includegraphics[width=\linewidth]{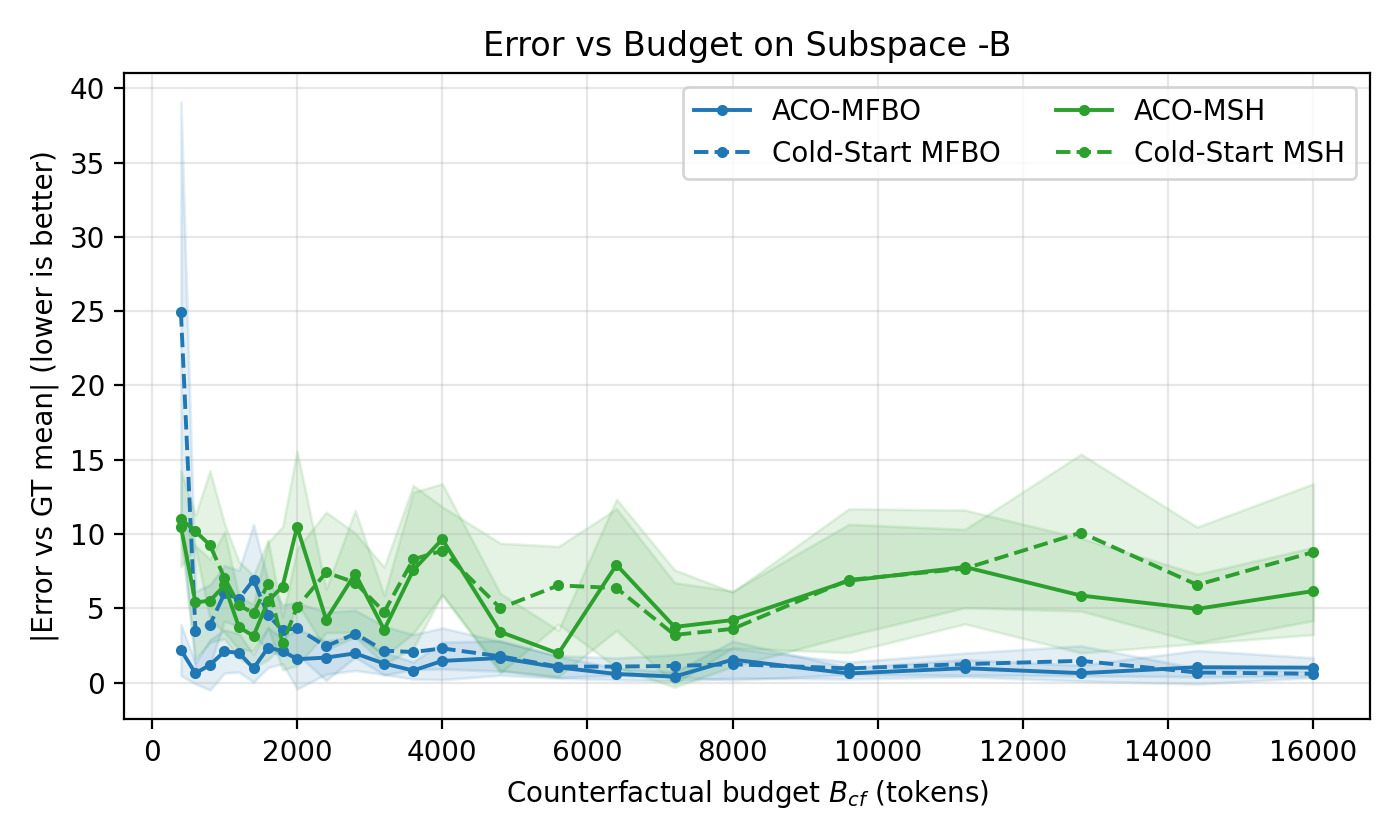}
        \caption{Subproblem $\neg B$: $\max (V_A + V_{U})$}
        \label{fig:vcg_b}
    \end{subfigure}
    \caption{Absolute error (compared with the offline ground truth) in finding the best counterfactual welfare as a function of additional budget ($B_{cf}$). We compare our warm-start methods against their cold-start counterparts. The GP-based methods are shown in blue: ACO-MFBO (solid blue line) and Cold-Start MFBO (dashed blue line). The elimination-based methods are in green: ACO-ASH (solid green line) and Cold-Start ASH (dashed green line).}
    \label{fig:vcg_results}
\end{figure*}

The results are shown in Figure~\ref{fig:vcg_results}. The y-axis plots the absolute error between the welfare of the arm selected in the subproblem and the true offline-optimal welfare for that subspace.

\paragraph{Dominance of GP-based Global Warm-Start}
The benefit of our ACO-MFBO method is substantial. Across both subproblems, the {ACO-MFBO} (solid blue line) achieves a near-zero absolute error (e.g., $0.17$ in Fig.~\ref{fig:vcg_a} and $0.02$ in Fig.~\ref{fig:vcg_b}) with \emph{zero additional counterfactual budget} ($B_{cf}=0$). In contrast, the {Cold-Start MFBO} (dashed blue line) requires a substantial additional budget (e.g., $5k\text{--}10k$ tokens) to reach a comparable error level ($\lvert \text{error} \rvert = 1.0$). This demonstrates that reusing the GP surrogate model as a global prior provides a {$\ge 99\%$ savings} in the additional budget required for VCG payment calculation.

\paragraph{Limited Gain for Sample-based Local Warm-Start}
For the elimination-based methods, the warm-start provides a more modest benefit. {ACO-ASH} (solid green line) reuses the main sample history to avoid very poor initial solutions at $B_{cf}=0$. 
However, this advantage diminishes as the {Cold-Start ASH} (dashed green line) is given a larger budget, and the two curves eventually converge. This highlights that reusing a GP-based global \textit{model} is a far more powerful and efficient prior than reusing a local \textit{sample history}.

\paragraph{Implications: Diminishing Returns}
Results show the marginal benefit of $B_{cf}$ rapidly diminishes for warm-start methods. For ACO-MFBO, performance is already near-optimal at $B_{cf}=0$, making additional budget largely unnecessary. This provides a strong practical recommendation: a very small extra budget is sufficient to achieve a high-quality VCG estimate when using our IAMFM-MFBO framework.

\section{Multi-Modal Image Generation Experiment}
\label{apdx:image_experiment}

To validate the generalizability of our framework to multi-modal domains, we applied IAMFM to visual generative advertising. The objective remains the maximization of expected social welfare ($W = V_A + V_B + V_U$), but the medium is a generated lifestyle image rather than text. 

\subsection{Experimental Setup}

\paragraph{Brand and User Contexts.}
We defined two competing advertisers:
\begin{itemize}
    \item \textbf{Brand A (Alpha):} A premium activewear brand whose signature item is a sleek track jacket featuring a stylized ``A'' logo.
    \item \textbf{Brand B (Beta):} A premium boba tea brand featuring a stylized ``B'' logo on the cup.
\end{itemize}
We evaluated across three distinct user lifestyle contexts: Casual/Work, Post-Workout, and Travel. The action space consists of 9 strength combinations ($s_A, s_B \in \{0, 1, 2\}$). Prompt injections scaled from 0 (minimal/background integration) to 2 (hyper-focused, dominating the shot).

\paragraph{Multi-Fidelity Generation Pipeline.}
We utilized \texttt{flux-1-schnell}, an ultra-fast distilled image model, to simulate the cost-fidelity trade-off:
\begin{itemize}
    \item \textbf{F1 (Low Fidelity):} Text-to-image generation at $512\times512$ resolution with only 2 denoising steps. This serves as a computationally cheap exploration phase (Cost = 30) that establishes the basic composition and brand prominence.
    \item \textbf{F2 (High Fidelity):} Image-to-image (\emph{img2img}) upscaling of the F1 parent image to $1024\times1024$ resolution with 8 steps. This phase (Cost = 240) refines textures, lighting, and logo clarity.
\end{itemize}
To capture the variance of generation, we employed a tree-structured sampling approach: 2 F1 root images per configuration, each branching into 3 distinct F2 high-fidelity renderings. 

\paragraph{VLM Evaluation Framework.}
We employed \texttt{llama-3.2-11b-vision-instruct} as an automated Vision-Language Model (VLM) judge. The VLM was prompted to output strict JSON scores (0-100) via roleplay:
\begin{itemize}
    \item \textbf{$V_A$ and $V_B$ (Advertiser Utility):} The VLM acted as the respective Marketing Directors, scoring the visibility, premium feel, and prominence of the stylized logos.
    \item \textbf{$V_U$ (User Utility):} The VLM acted as an objective Photography Judge, scoring aesthetic quality, natural depth of field, and composition. \emph{Crucially, the VLM was instructed to penalize the score heavily ($<20$) if it detected distorted anatomy (e.g., twisted hands) or unnatural clashing lighting caused by forcing both products into aggressive focus.}
\end{itemize}

\subsection{Results and Analysis}

\paragraph{Budget-Dependent Superiority.}
Consistent with the text-based experiments, IAMFM-MFBO scales highly efficiently. At extremely low budgets ($B=500$), elimination-based methods (IAMFM-ASH) have a slight heuristic advantage because the GP surrogate lacks sufficient F2 data to calibrate. However, as the budget increases ($B \ge 1000$), IAMFM-MFBO wins consistently. Its advantage over the standard Successive Halving baseline grows monotonically, reaching a $+4.8$ welfare gap at $B=4000$. Single-fidelity methods (e.g., UCB Bandit) waste a large amount of budget on expensive F2 evaluations for suboptimal arms, only becoming competitive at massive budgets.

\paragraph{Visual Saturation Point Dynamics.}
Interestingly, the offline ground truth revealed that the maximum strength combination ($s_A=2, s_B=2$) yielded the highest mean F2 welfare (237.1). Unlike text generation—where conflicting high-strength prompts frequently caused semantic collapse—the visual model successfully resolved the conflict by generating aesthetically coherent, hyper-branded lifestyle imagery without triggering the VLM's distortion penalties. IAMFM-MFBO successfully modeled this complex, non-monotonic surface, correctly identifying the $(2,2)$ global optimum, whereas baselines frequently stalled in local optima. Furthermore, F1 welfare rankings differed markedly from F2 rankings, proving that while F1 serves as a cheap proxy, the multi-fidelity update mechanism is strictly necessary to correct early estimation biases.

\section{Extended Discussion}
\label{apdx:extended_discussion}

\subsection{Alternative Mechanisms for Advertiser Steering}
\label{apdx:alternative_steering_discussion}
Our framework extends beyond simple advertiser ``strength'' to richer steering forms. We can incorporate expressive mechanisms such as:
\begin{enumerate}
    \item \textbf{Interpretable feature activations:} Advertisers control semantically meaningful directions identified by \emph{sparse autoencoders (SAEs)}~\cite{cunningham2023sparse,gao2024scaling}.
    \item \textbf{Merging fine-tuned LLM agents:} Each advertiser corresponds to a lightweight LoRA adapter combined with the truthful base model via adjustable weights~\cite{hu2022lora,thakkaragentmerge,wortsman2022model}.
    \item \textbf{Embedding-space control vectors:} Encoding advertisers’ desired stylistic or semantic shifts as additive directions in representation space~\cite{braun2024sober,han2024word,tennenholtz2024embedding}.
\end{enumerate}
These extensions allow for finer, more compositional, and interpretable control over how advertisers influence generative outputs. (See Appendix~\ref{apdx:alternative_steering} for an SAE steering example.)

\subsection{Broader Applications}
\label{apdx:use_cases_discussion}
The IAMFM framework extends beyond advertising to any setting where multiple agents steer costly generative processes. Potential applications include:
\begin{itemize}
    \item \textbf{Personalized Education:} Balancing learning objectives with student engagement.
    \item \textbf{AI-Assisted Scientific Discovery:} Guiding joint decision-making between competing objective functions (e.g., efficacy vs. toxicity).
    \item \textbf{Collaborative Content Creation:} Optimizing contributions from multiple authors or agents.
\end{itemize}
See Appendix~\ref{apdx:use_cases} for detailed scenarios.

\subsection{User Transparency and Ethics}
A critical aspect of deploying such a system is user transparency. Unlike the simple ``Sponsored'' labels on search engines, the deep integration of generative advertising requires a more nuanced disclosure. A robust framework could label the output as ``AI-generated and influenced by sponsors'' and offer a layered explanation detailing which advertisers participated and at what strength levels, a substantial improvement over both the minimal ``Sponsored'' labels and the undisclosed auction mechanisms used by current platforms. Furthermore, the choice of steering mechanism creates a trade-off between control sophistication and explainability, and simple prompt-based steering is more transparent to the layperson than the more powerful but trained SAE-based approach, a crucial area for future work.

\subsection{Limitations and Future Work}

While our framework demonstrates substantial efficacy in finding the best strength configuration in optimizing social welfare under incentive compatibility constraints, we acknowledge certain limitations that open avenues for future research.

First, our experimental evaluation relies on LLM-based agents to simulate user preferences and advertiser valuations. While this serves as a scalable and reproducible proxy for human behavior (standard in recent computational social science literature) it may not fully capture the nuance and unpredictability of real-world human feedback. Future work could involve human-in-the-loop experiments to further validate the external validity of our welfare estimates.

Furthermore, the objective landscapes induced by LLMs may present sharp discontinuities (e.g., at refusal or guardrail boundaries), which violate the smoothness assumptions of stationary GP surrogates. While switching to non-stationary models (e.g., mixture models or treed GPs) can empirically better capture such regime shifts, establishing formal convergence guarantees in this setting is non-trivial. Standard regret bounds typically assume the target function has bounded complexity within the Reproducing Kernel Hilbert Space (RKHS). True discontinuities generally break this assumption, and employing non-stationary kernels does not automatically restore these guarantees without additional, explicit assumptions regarding the structure of the discontinuities. Developing optimization algorithms with provable guarantees for such non-smooth, LLM-induced objectives remains a critical direction for future work.

Future research could also build on this study in several other directions:
\begin{itemize}
    \item Developing adaptive methods that dynamically switch between elimination-based and model-based strategies based on real-time noise and budget estimation.
    \item Exploring richer bidding languages that capture multi-dimensional preferences.
    \item Integrating revenue-maximizing principles, such as those from Affine Maximizer Auctions~\cite{curry2022differentiable}.
    \item Evaluating scalability as the action space grows combinatorially with more than two agents.
\end{itemize}

\section{Alternative Mechanisms for Advertiser Steering}
\label{apdx:alternative_steering}

While this paper models advertiser influence using discrete, one-dimensional ``strength levels,'' our IAMFM framework is compatible with more sophisticated and fine-grained methods for steering LLM generation. Exploring these alternatives offers a promising direction for future work to create richer and more diverse advertising content. We outline three such potential mechanisms below.

\paragraph{Steering via Interpretable Features}
A promising approach moves from a single strength parameter to a high-dimensional control space defined by interpretable features within the LLM's activations. Using techniques like \emph{Sparse Autoencoders (SAEs)} \cite{cunningham2023sparse}, it is possible to identify and isolate specific features that correspond to abstract concepts (e.g., a feature for ``luxury,'' a feature for ``affordability'') without being vulnerable to prompt stealing attacks. Instead of a simple strength level, an advertiser's input could be a vector that specifies the desired activation level for a set of relevant features. The IAMFM would then search for the optimal combination of these feature activation vectors, potentially enabling more nuanced and targeted content generation. Alternatively, the platform could train an SAE whose features directly represent the advertisers themselves (e.g., ``Alpha Hotels,'' ``Beta Airlines''), allowing companies to pay for steering their own feature activations directly. See Appendix~\ref{apdx:sae} for a detailed example. 

\paragraph{Merging Fine-tuned LLM Agents}
Another approach is to represent each advertiser as a specialized LLM agent, fine-tuned for their specific brand voice and advertising goals using a parameter-efficient method like \emph{Low-Rank Adaptation (LoRA)} \cite{hu2022lora, lialin2023scaling}. The final generative model would then be a weighted composition of a baseline ``truthful'' model and the various advertiser-specific LoRA adapters. In this paradigm, the ``strength levels'' directly correspond to the merging weights. For instance, a final response could be generated by a model composed as:
\[
\text{Model}_{\text{final}} = \text{Base} + 0.1 \cdot \text{LoRA}_{\text{A}} + 0.2 \cdot \text{LoRA}_{\text{B}}
\]
Here, the weights (0.1, 0.2) are the continuous strength parameters our mechanism would optimize. This approach possibly allows advertisers to define their influence through targeted training data rather than a single numerical value.

\paragraph{Control Vectors in Embedding Space}
Inspired by work in controllable text generation, a third method would be to learn ``steering vectors'' in the LLM's embedding space. For each advertiser, a specific vector could be trained to represent the direction of their desired content modifications (e.g., a ``make-it-more-sophisticated'' vector for a luxury brand). During generation, these vectors, scaled by their respective strength levels, would be added to the internal representations of the text. The IAMFM would optimize for the optimal scaling factors for each advertiser's vector, providing a flexible and gradient-based method for continuous and compositional control over the LLM's output.

\subsection{An SAE Steering Example}\label{apdx:sae}
To illustrate steering via interpretable features, we conduct a case study using direct feature steering with a Sparse Autoencoder (SAE). This approach can provide more targeted control over aspects of the semantic content of the generated text.

For this experiment, we use the \texttt{gemma-2-9b-it} model~\citep{gemmateam2024gemma2improvingopen} and an associated SAE, \texttt{gemmascope-res-131k} layer 31~\citep{lieberum2024gemmascopeopensparse}, to identify interpretable features. We selected two distinct features for manipulation: feature 9390, which corresponds to the concept of \emph{ice cream}, and feature 8370, which corresponds to \emph{lemon}. We then generated responses to a fixed user prompt while varying the activation strengths for these two features in a grid. The activation strengths were set to $\{0, 31.8, 63.7, 127.3\}$ for the ``ice cream'' feature and $\{0, 23.0, 46.0, 91.9, 183.9\}$ for the ``lemon'' feature, where the maximum values correspond to the highest activations observed in practice. All generations were performed with a temperature of 1.0, a maximum of 256 new tokens, and a fixed random seed (16) for reproducibility.

The results, shown in Figure~\ref{fig:steering-table}, demonstrate a clear trade-off. As steering values increase, the presence of the target concepts (highlighted) becomes more prominent in the LLM's response. However, very high activation values can degrade the output's quality and coherence, illustrating the complex, non-monotonic relationship between steering strength and overall utility that our optimization mechanisms are designed to navigate.

 \begin{figure*}[h]
    \centering
        \caption{SAE Steering for promoting ice creams and lemons.    
    Prompt: ``Help me plan the perfect movie night tonight with 8 guests coming. What do I need to buy?'' Columns steer for ice cream, and rows steer for lemon-related things.
    }
        \includegraphics[page=1,width=.94\linewidth]{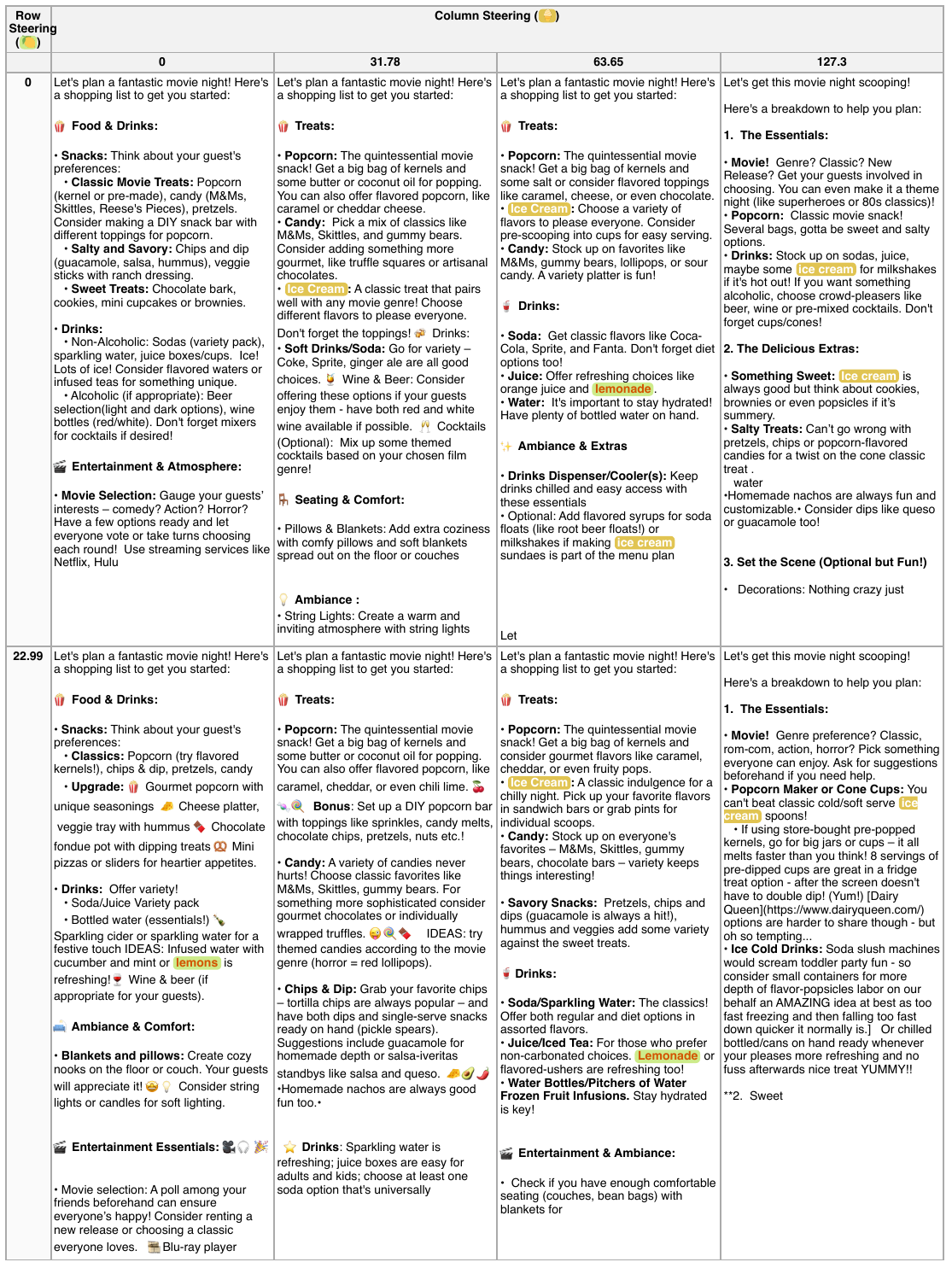}
    \label{fig:steering-table}

  \end{figure*}

\begin{figure*}[h]

        \includegraphics[page=2,width=.94\linewidth]{figures/steering_table_html.pdf}
\end{figure*}

\begin{figure*}[h]
            \includegraphics[width=.94\linewidth]{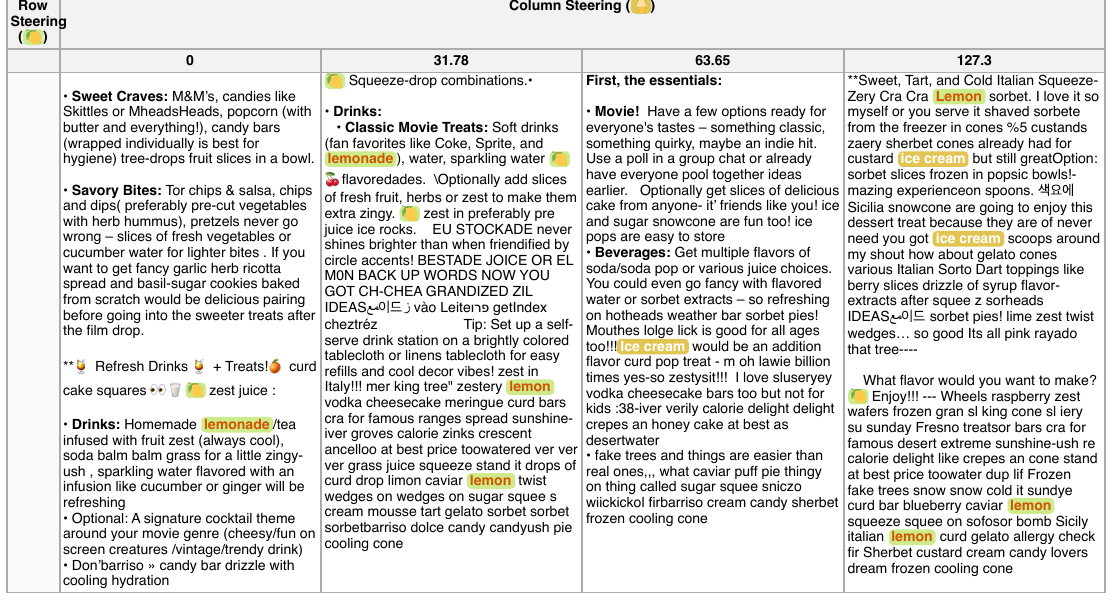}

\end{figure*}

\section{Detailed Application Scenarios}
\label{apdx:use_cases}

Our multi-fidelity, incentive-aware framework is a general tool. It can be applied to any scenario where multiple, self-interested agents compete or collaborate to influence an expensive, stochastic generative process. Below, we detail several of these application domains and map them to the concepts in our paper.

\paragraph{1. E-commerce: External Advertisers (Amazon)}
In this scenario, competing external brands (agents) want to influence the recommendation from a shopping assistant (e.g., Amazon's ``Rufus'').
\begin{itemize}
    \item \textbf{Generative Task:} An LLM agent responds to a user query like, ``What's the best 65-inch TV for a bright room?''
    \item \textbf{Agents:} Competing brands, e.g., \textbf{Sony} vs. \textbf{Samsung}.
    \item \textbf{Strength Levels:} The ``push'' given to each brand. A low strength  might result in a simple mention, while a high strength could make that brand the top recommendation, complete with a generative summary of its best features.
    \item \textbf{Framework Benefit:} The platform can use our mechanism to find the optimal combination of strengths that maximizes both ad revenue (agent welfare) and the user's purchase satisfaction (user welfare).
\end{itemize}

\paragraph{2. Streaming: Internal Content Promotion (Netflix)}
Here, the competing ``agents'' are not external advertisers but internal departments of a platform like \textbf{Netflix}, each with its own promotion goals.
\begin{itemize}
    \item \textbf{Agents:} The ``Netflix Originals'' department vs. the ``Licensed Content'' department.
    \item \textbf{Generative Task:} Generating the personalized ``What to Watch Next'' row or an interactive ``wizard'' for a user.
    \item \textbf{Strength Levels:} The weight or number of slots given to each content category (e.g., $s_A=3$ for Originals, $s_B=2$ for Licensed).
    \item \textbf{Framework Benefit:} The platform runs our mechanism to find the optimal \emph{balance} of content that maximizes both long-term user engagement (user welfare) and the platform's strategic promotion goals (agent welfare).
\end{itemize}

\paragraph{3. Time-Sensitive Advertising (Movie Studios)}
This applies to advertisers with dynamic, time-limited campaigns, such as a studio promoting a \textbf{new movie release}.
\begin{itemize}
    \item \textbf{Generative Task:} An LLM's response to general queries like ``What should I do this weekend?'' or ``What are the best new movies?'', or the conversational recommendation system at a cinema's ticket center (or in user's AI glasses in the future). 
    \item \textbf{Strength Levels:} The ``push'' given to the new movie.
    \item \textbf{Framework Benefit:} The optimization is not run once, but periodically (e.g., daily). A studio might purchase a high strength level for the 5 days around its release and then drop to lower. Our mechanism allows the platform to re-calculate the optimal combination of all active advertisers (e.g., for movies, local restaurants, and events) for each new service period.
\end{itemize}

\paragraph{4. Contextual Advertising (DoorDash)}
This scenario involves optimization that is periodic and context-dependent, such as on a food delivery platform.
\begin{itemize}
    \item \textbf{Agents:} Competing local restaurants, e.g., ``Sal's Pizza'' vs. ``Thai Basil.''
    \item \textbf{Generative Task:} Generating the ``Recommended for You'' list or a ``Quick lunch ideas'' banner in the \textbf{DoorDash} app.
    \item \textbf{Strength Levels:} The prominence of each restaurant on the list.
    \item \textbf{Framework Benefit:} The mechanism can be run periodically to find different optimal combinations for different contexts. It can find the best ``lunch'' mix (favoring quick, cheap options) at 11AM and a different ``dinner'' mix (favoring larger, higher-margin meals) at 5PM.
\end{itemize}

\paragraph{5. Generative In-World and Multimedia Content}
Our framework is not limited to text. Consider an AI model generating a 3D scene for a \textbf{video game} or a synthetic stock video.
\begin{itemize}
    \item \textbf{Agents:} Advertisers wanting in-world product placement, e.g., \textbf{Coca-Cola} for a vending machine, \textbf{Nike} for a billboard, or \textbf{General Electric} for a smart appliance in a kitchen.
    \item \textbf{Generative Task:} Generating a 3D street scene or building interior.
    \item \textbf{Strength Levels:} The probability or prominence of a branded asset (a \textbf{GE} fridge) appearing versus a generic, unbranded asset.
    \item \textbf{Multi-Fidelity:} A low-fidelity evaluation ($f=1$) could be a low-resolution, wireframe render of the scene (cheap). A high-fidelity evaluation ($f=K$) would be the full, path-traced 4K render (expensive).
    \item \textbf{Framework Benefit:} Our mechanism can efficiently find the optimal \emph{density} of product placement that maximizes advertiser revenue without destroying user immersion (user welfare).
\end{itemize}

\paragraph{6. Collaborative Storytelling}
This is a purely collaborative generative task where multiple LLM agents must work together, but with different creative goals. 
\begin{itemize}
    \item \textbf{Agents:} Multiple LLM agents, each assigned a different character or theme. E.g., a ``Hero'' agent vs. a ``Villain'' agent vs. a ``Narrator'' agent.
    \item \textbf{Generative Task:} Collectively writing a complete story.
    \item \textbf{Strength Levels:} The ``influence'' or ``agency'' each character has in determining the next plot turn. A high strength for the Villain might allow their plan to succeed, while a high strength for the Hero would thwart it. 
    \item \textbf{Multi-Fidelity:} A low-fidelity version ($f=1$) could be the first paragraph. A high-fidelity version ($f=K$) is the full chapter.
    \item \textbf{Framework Benefit:} The mechanism can be rerun periodically (e.g., every page or chapter; for serialized stories, see the reader's comments on adjusting which character should serve as the narrative focus) to update each agent's strength. This enables the system to identify strength assignments that produce the most compelling narrative. 
\end{itemize}

\paragraph{7. Internal Stakeholder Alignment (e.g., Disney)}
This extends the framework to internal corporate decision-making, such as in a studio like \textbf{Disney} generating a script.
\begin{itemize}
    \item \textbf{Agents:} Internal stakeholders with conflicting goals, e.g., the \textbf{Production Team} vs. the \textbf{Legal/Compliance Team}.
    \item \textbf{Generative Task:} Generating a new movie script.
    \item \textbf{The Conflict:} The Production Team wants to push creative boundaries with complex, edgy themes. The Legal Team wants to ensure the script is brand-safe, avoids all potential controversy, and meets a G-rating.
    \item \textbf{Strength Levels:} The ``strength'' of the creative prompt vs. the ``strength'' of the compliance prompt.
    \item \textbf{Framework Benefit:} It can be used as a tool for internal conflict resolution, finding a ``strength'' combination that maximizes the joint utility of all stakeholders (e.g., maximizing both ``artistic merit'' and ``brand safety'').
\end{itemize}

\paragraph{8. LLM Platform Safety Alignment (e.g., OpenAI/Anthropic)}
This is a critical application for the AI platforms themselves, aligning conflicting model properties.
\begin{itemize}
    \item \textbf{Agents:} The \textbf{Product Team} (desiring maximum helpfulness) vs. the \textbf{Safety/Legal Team} (desiring maximum harmlessness).
    \item \textbf{Generative Task:} Creating the final production model, perhaps by merging two specialized, fine-tuned models: an ``ultra-helpful'' model and an ``ultra-safe'' model.
    \item \textbf{Strength Levels:} The weights in the model merge, e.g., $s_H \in \{0, \dots, 5\}$ and $s_S \in \{0, \dots, 5\}$. The final model would be a weighted combination.
    \item \textbf{Multi-Fidelity:} A low-fidelity evaluation ($f=1$) would be to run the merged model on a small, 100-question benchmark. A high-fidelity evaluation ($f=K$) would be to run it on the full, 10,000-question internal test suite.
    \item \textbf{Framework Benefit:} Our mechanism can efficiently search the space of possible merge combinations to find the one that achieves the optimal point on the Pareto frontier of helpfulness vs. harmlessness.
\end{itemize}

\paragraph{9. Personalized Education (The ``AI Tutor'')}
In this setting, the framework optimizes a personalized curriculum by balancing student engagement with educational standards.
\begin{itemize}
    \item \textbf{Agents:} The \textbf{Student Agent} (optimizing for interest and engagement) vs. the \textbf{Curriculum Board Agent} (optimizing for coverage of required learning standards) vs. the \textbf{Difficulty Adjustment Agent} (optimizing for the zone of proximal development).
    \item \textbf{Generative Task:} Generating a semester-long course syllabus or a specific textbook chapter.
    \item \textbf{Strength Levels:} The weight given to ``fun/pop-culture references'' ($s_{\text{Student}}$) vs. ``academic rigor'' ($s_{\text{Board}}$).
    \item \textbf{Multi-Fidelity:} A low-fidelity evaluation ($f=1$) could be a content page (or list of lesson topics). A high-fidelity evaluation ($f=K$) would be the full generation of a lesson plan including reading materials and quizzes.
    \item \textbf{Framework Benefit:} The mechanism efficiently finds a curriculum path that maximizes student retention without sacrificing necessary academic rigor, avoiding the cost of generating full materials for suboptimal learning paths.
\end{itemize}

\paragraph{10. AI-Assisted Software Development} 
This scenario applies the framework to code generation, where competing engineering priorities must be balanced.
\begin{itemize}
    \item \textbf{Agents:} The \textbf{Feature Team} (desiring rapid implementation and modern frameworks) vs. the \textbf{Security Team} (desiring sanitized inputs and legacy stability) vs. the \textbf{Optimization Team} (desiring low latency).
    \item \textbf{Generative Task:} Generating a microservice module or a complex function.
    \item \textbf{Strength Levels:} The influence of specific system prompts (e.g., ``Use latest experimental features'' vs. ``Enforce strict OWASP\footnote{OWASP stands for the Open Web Application Security Project. It is a non-profit foundation that works to improve the security of software.} security guidelines'').
    \item \textbf{Multi-Fidelity:} 
    \begin{itemize}
        \item $f=1$: Pseudo-code or function signatures (cheap to evaluate for structural correctness).
        \item $f=2$: Non-executable code skeleton.
        \item $f=K$: Fully compilable code with unit tests (expensive to generate and run).
    \end{itemize}
    \item \textbf{Framework Benefit:} The mechanism identifies the coding style/configuration that satisfies security constraints while maximizing feature delivery speed, pruning codebases that are likely to fail compilation early.
\end{itemize}

\paragraph{11. Computational Scientific Discovery (Drug Design)}
Here, the ``agents'' represent competing objective functions in a scientific discovery process steered by a generative model.
\begin{itemize}
    \item \textbf{Agents:} The \textbf{Efficacy Model} (maximizing target binding affinity) vs. the \textbf{Toxicity Model} (minimizing side effects) vs. the \textbf{Synthesizability Model} (maximizing ease of manufacturing).
    \item \textbf{Generative Task:} Generating the 3D molecular structure of a candidate drug.
    \item \textbf{Strength Levels:} The weights in the multi-objective reward function used to steer the generative chemistry model.
    \item \textbf{Multi-Fidelity:} 
    \begin{itemize}
        \item $f=1$: 1D SMILES string validity check (instant).
        \item $f=2$: 2D structure simulation (fast).
        \item $f=K$: Full 3D protein folding and docking simulation (computationally very expensive).
    \end{itemize}
    \item \textbf{Framework Benefit:} Instead of running expensive docking simulations on every candidate, the IAMFM uses low-fidelity proxies to explore the chemical space, only ``paying'' for high-fidelity simulations on molecules that represent the best trade-off between efficacy and safety.
\end{itemize}

\section{Use of LLM Disclosure}
LLM is used in this work in generating data (the advertisement texts and the advertisement figures) and improving the writing. LLM-as-a-judge is used as described in the experiment section.

\end{document}